\documentclass[journal]{IEEEtran}
\usepackage{amsmath}
\usepackage{amssymb}

\usepackage{amsthm}
\usepackage{framed}
\usepackage{braket}

\def\II{\mathbb I}

\newtheorem{Lmm}{Lemma}
\newtheorem{Thm}{Theorem}
\newtheorem{Dfn}{Definition}
\newtheorem{Crl}{Corollary}

\newtheorem{proposition}{Proposition}

\begin{document}

\title{Leftover hashing from quantum error correction:\\
Unifying the two approaches to the security proof of quantum key distribution
}
\author{Toyohiro~Tsurumaru%,~\IEEEmembership{Nonmember}
\thanks{The author is with Mitsubishi Electric Corporation, Information Technology R\&D Center,
5-1-1 Ofuna, Kamakura-shi, Kanagawa, 247-8501, Japan (e-mail: Tsurumaru.Toyohiro@da.MitsubishiElectric.co.jp).
This paper was presented in part at the 38th Quantum Information Technology Symposium (QIT38), Hiroshima, Japan, June 4-5, 2018; also in part at 2019 Symposium on Cryptography and Information Security (SCIS2019), Otsu, Japan, Jan. 22-25, 2019; and also in part at the 9th International Conference on Quantum Cryptography (QCrypt 2019), Montreal, Canada, Aug. 26-30,  2019.
}
}

\maketitle

\begin{abstract}
We show that the Mayers-Shor-Preskill approach and Renner's approach to proving the security of quantum key distribution (QKD) are essentially the same.
We begin our analysis by considering a special case of QKD called privacy amplification (PA).
PA itself is an important building block of cryptography, both classical and quantum.
The standard theoretical tool used for its security proof is called the leftover hashing lemma (LHL).
We present a direct connection between the LHL and the coding theorem of a certain quantum error correction code.
Then we apply this result to proving the equivalence between the two approaches to proving the security of QKD.
\end{abstract}


\section{Introduction}
\label{sec:introduction}
Quantum key distribution (QKD) is a technique for distributing information-theoretically secure keys between two remote parties connected by a quantum channel \cite{BB84, Nielsen-Chuang}.
Today QKD is a real-world technology; 
tests have been carried out in various realistic situations, including metropolitan fiber networks \cite{Sasaki:11} and satellite communication \cite{PhysRevLett.120.030501}; there are also commercially available products \footnote{See, e.g., the webpages of id Quantique SA and of Quantum CTek Group.}.
Underlying this technology are information-theoretic proofs of the security of QKD.
No QKD system can be secure unless a rigorous security proof is given.

At the present, there are two major approaches for obtaining security proofs:
\begin{itemize}
\item Phase error correction (PEC)-based approach:
This is usually called the Shor-Preskill approach, or sometimes the Mayers-Shor-Preskill approach \cite{Mayers98,SP00}.
One constructs a virtual QKD protocol which has a quantum phase error correction (PEC) algorithm embedded inside.
The security analysis is then reduced to upper bounding the failure probability of the quantum PEC, e.g., by using a coding theorem.
This approach was initiated by Mayers \cite{Mayers98}, and later simplified and improved by many others including Shor and Preskill \cite{SP00}, Koashi \cite{Koashi, KoashiNJP09}, and Hayashi \cite{H07}.
\item Leftover hashing lemma (LHL)-based approach:
This is usually called Renner's approach \cite{RennerPhD}.
Once a lower bound on the (smooth) conditional min-entropy of a sifted key is estimated, one can guarantee the secrecy of the corresponding final key, simply by applying an existing formula called the leftover hashing lemma (LHL, or lemma \ref{lmm:leftover_hashing_lemma} of this paper).
\end{itemize}
As far as we know, the mathematical relations between these two methods were investigated only in restricted cases \cite{PhysRevA.78.032335,doi:10.1098/rspa.2010.0445,5967913}, and because of that, these methods have generally been considered independent of each other.
It now seems customary to publish two papers that prove the security of the same QKD protocol by using each of the two approaches: E.g., to name a few, for the asymptotic case of Bennett-Brassard 1984 (BB84) protocol \cite{BB84} there are papers of the PEC-based approach \cite{Mayers98, SP00} and of the LHL-based \cite{RennerPhD}; for the finite-size case there are refs. \cite{H07} and \cite{TLGR12}; for the Bennett 1992 (B92) protocol \cite{B92} there are refs. \cite{TKM03} and \cite{RennerPhD}.
Empirically, however, obtained results are the same for most problems of practical interest (e.g., the final key rate for a given value of the security parameter $\varepsilon$), regardless of the approach.

The goal of this paper is to show that these two approaches are in fact essentially the same.
We demonstrate this fact by presenting a direct connection between these two approaches.
That is, for QKD protocols in general, we present an explicit procedure for converting a security proof of one approach to that of another approach.
We also show that the two approaches achieve the same level of security; i.e., they give the same security bound except for the presence of an inessential constant factor.
Thus hereafter, there is no need to publish two security proofs for the same QKD protocol; one is enough.

We begin our analysis by considering privacy amplification (PA) algorithms \cite{476316}.
PA can be thought of as a special case of QKD where the sender Alice alone generates a secret key in the absence of the receiver Bob.
At the same time, PA itself is an important building block of cryptography, both classical and quantum (see, e.g., \cite{476316,Impagliazzo:1989:PGO:73007.73009,tuyls2007security}).
The LHL mentioned above is the standard theoretical tool used for guaranteeing its security \cite{RennerPhD,476316,Impagliazzo:1989:PGO:73007.73009,5961850}.
For this algorithm, we show that (i) there exists a direct connection between a PA algorithm and a certain quantum phase error correction (PEC) algorithm, and (ii) the LHL follows from a coding theorem of the quantum PEC algorithm.

Then we generalize these results on PA to QKD, and show that the same connection exists also for QKD protocols in general.
That is, for QKD protocols in general, we present an explicit procedure for converting a security proof of the LHL-based approach to that of a PEC-based approach, and vice versa.
We also show that the obtained security bounds are the same, except for the presence of an inessential constant factor.

We close this section by giving very technical remarks.
The first remark is that, throughout the paper, we assume that one always uses the universal$_2$ hash function for PA, despite the fact that there are other functions suitable for PA, e.g., the {\it almost dual} universal$_2$ hash function \cite{FS08,TH13,HT14}, Trevisan's extractor \cite{Trevisan,DPVR12}, etc. (see, e.g., \cite{HT14}). The reasons for this restriction are as follows.
First, as long as the security of a given protocol is concerned, one does not benefit from using functions other than the universal$_2$.
That is, using those functions do not result in a longer final key length than the universal$_2$.
Rather, their advantage is in performance of the implementation; e.g., they are faster on software, or require shorter public nonce (see, e.g., \cite{HT14}).
Second, the universal$_2$ hash function is widely used in classical cryptography \cite{tuyls2007security}, quantum cryptography \cite{Sasaki:11,PhysRevLett.120.030501,VanAssche}, and quantum random number generation \cite{MYCQZ16}.
In particular, it is the most used function in quantum cryptography \cite{Sasaki:11,PhysRevLett.120.030501,VanAssche}.
%In short, if one is concerned about the security of a given protocol, it suffices to prove it for a particular type of functions, and we choose the universal$_2$ function since it is the most important.

The second remark is that similar results for limited cases were obtained by the present author and a co-worker \cite{TH13}. Namely, the conversion from the PEC-based approach to the LHL-based approach has been demonstrated for the case where one uses the {\it almost dual} universal$_2$ hash function (which includes the universal$_2$ as a special case) in the BB84 protocol or in a quantum wiretapping with the Pauli channel \cite{TH13}; but the case of QKD protocols in general and the converse (i.e. the LHL-to-PEC conversion) were not proved.
This paper presents the conversions of both directions (i.e., the equivalence) for QKD protocols in general under the condition that the universal$_2$ hash function is used for PA.
The equivalence for other functions, e.g., the {\it almost dual} universal$_2$ hash function or Trevisan's extractor, still remains an open problem.

The third remark is on the relationship between the results obtained here and the previous papers by Renes and his collaborators \cite{PhysRevA.78.032335,doi:10.1098/rspa.2010.0445,5967913}.
These papers indeed consider settings similar to ours, and also share the same motivation as ours, of investigating the relationship between the PEC-based and the LHL-based approaches.
However, unlike in the present paper, they discuss certain limited cases only, and the results are not sufficient to demonstrate the equivalence of the two approaches.
More specifically, in refs. \cite{PhysRevA.78.032335,5967913}, they consider the idealized setting where Eve performs collective attacks, and only show that the rate of secret fraction are the same in the two approaches, in the asymptotic limit where the number of blocks goes to infinity, with the security parameter $\varepsilon$ being fixed.
Paper \cite{doi:10.1098/rspa.2010.0445} considers the coherent attack against a bit string of any finite length, which is more realistic, and proves essentially that if there exists a linear error correcting code with a desirable property, the PEC-based approach can be converted to the LHL-based approach. However, the paper discusses neither a practical hash function nor the conversion in the converse  (i.e. the LHL-to-PEC conversion) except for two extreme cases. Thus they do not prove the equivalence. In contrast, in the present paper, we consider the practical setting where Eve performs the coherent attack against a bit string of any finite length and the legitimate user(s) use the linear universal$_2$ hash function, and show the equivalence of the PEC-based and the LHL-based approaches.

\section{Review of Privacy Amplification}
Privacy amplification (PA) \cite{476316} is a technique for extracting a completely private bit string, from a given bit string which may be partially leaked outside.
This can also be viewed as a special case of QKD where Alice alone generates a final key in the absence of Bob.
In this section we review the definition and the security of PA.

\label{sec:privacy_amplification}
\subsection{Definition}
\label{sec:definition_of_PA}
%\subsubsection{Setting and procedure}
%\label{sec:procedure_and_output_of_PA}

\subsubsection{Initial state and algorithm}
PA starts with a sub-normalized classical-quantum (cq) state
\begin{equation}
\rho_{AE}^{\rm ini}=\sum_a|a\rangle\langle a|_A\otimes \rho^{{\rm ini},a}_E,
\label{eq:initial_state_def}
\end{equation}
where bit string $a\in \{0,1\}^n$ is owned by the legitimate user (say, Alice), and $\rho^{{\rm ini},a}_E$ by the eavesdropper Eve.
Note that Eve may be able to obtain a certain amount of information of string $a$ by measuring $\rho^{{\rm ini},a}_E$, unless $\rho^{{\rm ini},a}_E$ are all equal.

PA is a technique to extract from string $a$ a random string $k$ which Eve cannot guess.
More precisely, Alice selects a random function $g$, and then calculates a shorter string $k=g(a)$.
The idea is to realize the situation where Eve's state corresponding to $k$ ($\rho^{{\rm fin},g,k}$ of eq. (\ref{eq:final_state_rhogk}) below) are averaged, and less distinguishable than $\rho^{{\rm ini},a}_E$ in the initial (sub-normalized) state.
\begin{oframed}
\noindent\underline{{\bf Actual PA scheme} $(\Pi^{\rm ap},\rho^\text{ini-ap})$}

\medskip 

\noindent{\bf Initial state:} A sub-normalized cq state $\rho_{AE}^\text{ini-ap}$.

\medskip 

\noindent{\bf Algorithm} $\Pi^{\rm ap}_{AKG}$:
\begin{enumerate}
\item {\bf Choice of hash function}: Choose a random function $g$ and announce it publicly by writing it in a Hilbert space ${\cal H}_G$.
\item {\bf Actual PA using hash function $g$} ($\Pi^{{\rm ap},g}$):
\begin{enumerate}
\item Measure space ${\cal H}_A$ in the $Z$-basis to determine the value of $a$.
\item Calculate hash value $k=g(a)$, and store it in ${\cal H}_K$.
\end{enumerate}
\begin{itemize}
\item[---] We denote the final sub-normalized state by $\rho^\text{fin-ap}$

(=$\rho_{AEKG}^\text{fin-ap}=\Pi^{\rm ap}(\rho^\text{ini-ap})$).
\end{itemize}
\end{enumerate}
\end{oframed}

\subsubsection{Terminology}
We call a set of an algorithm and an initial (sub-normalized) state a {\it scheme}.

Whenever we say a function $g$ is {\it random}, (i) there are a predetermined set of functions ${\cal G}=\set{g}$ (also called function family ${\cal G}$) as well as a predetermined probability $p_G(g)$, and (ii) in an actual algorithm or protocol, a legitimate player, Alice or Bob, selects $g$ with probability $p_G(g)$.
In what follows, we often regard such $g$ as a random variable and denote it by capital $G$; in this case the probability is written as $p_G(g)=\Pr(G=g)$.
All random functions considered in this paper have their output shorter than their input, and for this reason we often call them random {\it hash} functions, and their output a {\it hash value}.

\subsubsection{Assumptions}
\label{sec:assumptions_PA}
Whenever we discuss a PA scheme, we impose the following assumptions.
As to Hilbert spaces,  ${\cal H}_A$ and ${\cal H}_K$ are under Alice's control, and ${\cal H}_E$ is under Eve's.
${\cal H}_G$ is a public space under both Alice's and Eve's control.
As to random hash functions, we assume:
\begin{itemize}
\item Functional forms of $g$ and probability $\Pr(G=g)$ are public.
\item Functions $g$ are of the form $g:\{0,1\}^n\to\{0,1\}^m$ with $n\ge m$, and are linear and surjective.
\end{itemize}
Note that, from the first item, it is clear that the randomness and the secrecy of $k$ are from those of the original string $a$, not of $g$.
We remark that the surjectivity of $g$ in the second item is solely for the sake of simplicity, and all the results of this paper in fact hold without it.

\subsection{Standard security criterion of PA ($\varepsilon$-secrecy)}
\label{sec:security_criterion_PA}

In this and the next subsections, we abbreviate $\rho^\text{ini-ap}$, $\rho^\text{fin-ap}$ as $\rho^\text{ini}$, $\rho^\text{fin}$.

In the security analysis, we may focus on system ${\cal H}_{KEG}$, because we are concerned with how much information of final key $k$ is accessible from Eve through her degrees of freedom ${\cal H}_{EG}$, and nothing else.
By construction of $\Pi^{\rm ap}$, it is obvious that the reduced sub-normalized state for ${\cal H}_{KEG}$ corresponding to $\rho^{\rm fin}$, i.e., $\rho_{KEG}^{\rm fin}={\rm Tr}_{A}(\rho_{KAEG}^{\rm fin})$, generally takes the form
\begin{eqnarray}
\rho_{KEG}^{\rm fin}&=&\sum_{g}\Pr(G=g)\,\rho^{{\rm fin},g}_{KE}\otimes |g\rangle\langle g|_G,
\label{eq:rho_KEG_definied}
\\
\rho_{KE}^{{\rm fin},g}&=&\sum_{k}\ket{k}\bra{k}_K\otimes \rho^{{\rm fin},g,k}_{E},
\label{eq:final_state}
\\
\rho^{{\rm fin},g,k}_{E}&=& \sum_{a\in g^{-1}(k)}\rho^{{\rm ini},a}_E.
\label{eq:final_state_rhogk}
\end{eqnarray}
Note that $\rho^{{\rm fin},g}$ can be interpreted either as the final sub-normalized state conditioned on Alice's choice $g$ of a hash function (as in eq. (\ref{eq:rho_KEG_definied})), or as the result of applying operation $\Pi^{\text{pa},g}$ (i.e., 2 with $g$ fixed) to the initial sub-normalized state; $\rho^{\rm fin}=\Pi^{\text{pa},g}(\rho^{\rm ini})$.

It is customary to measure the security of $\rho_{KEG}^{\rm fin}$ by comparing it with the ideal case.
To this end, we define the ideal sub-normalized state $\rho_{KEG}^{\rm ideal}$ corresponding to $\rho_{KEG}^{\rm fin}$ to be
\begin{eqnarray}
\rho_{KEG}^{\rm ideal}&:=&\sum_{g}\Pr(G=g)\,\rho_{KE}^{\rm ideal}\otimes |g\rangle\langle g|_G,\\
\rho_{KE}^{\rm ideal}&:=&2^{-m}\sum_{k}|k\rangle\langle k|_K\otimes \rho_E,
\label{eq:rho_KE_ideal_defined}
\\
\rho_E&=&{\rm Tr}_{KG}(\rho^{{\rm fin},g}_{KEG})={\rm Tr}_A(\rho^{\rm ini}_{AE})=\sum_a\rho^{{\rm ini},a}_{E}.
\label{eq:rho_E_defined}
\end{eqnarray}
This sub-normalized state indeed describes the ideal situation where Eve's sub-normalized states corresponding to $k$, i.e., $\rho^{{\rm ideal},k}_E=2^{-m}\rho_E$, are all equal, and thus all values of $k$ are equally probable to her.

We measure the security of $k$ against Eve by a distance between these actual and ideal (sub-normalized) states (see e.g. ref. \cite{RennerPhD}).
More specifically: We define the $L_1$ distance of two sub-normalized states $\mu,\nu$ to be the $L_1$ norm (or 1-norm) of their difference, $\|\mu-\nu\|_1$.
Then we define the quantity $d_1$ of the final (sub-normalized) state $\rho_{KEG}^{\rm fin}$ to be
\begin{eqnarray}
d_1(\rho_{KEG}^{\rm fin})&:=&\left\|\rho_{KEG}^{\rm fin}-\rho_{KEG}^{\rm ideal}\right\|_1,
\label{eq:average_trace_dist_defined}
\end{eqnarray}
i.e. the $L_1$-distance between $\rho_{KEG}^{\rm fin}$ and the corresponding ideal (sub-normalized) state $\rho_{KEG}^{\rm ideal}$.
A PA scheme is said to satisfy the $\varepsilon$-secrecy\footnote{In some literature (e.g. ref. \cite{Nielsen-Chuang}), the quantity $\frac12\|\mu-\nu\|_1$, one half of the $L_1$-distance, is called the trace distance.
In this terminology, the $\varepsilon$-secrecy means that trace distance between $\rho_{KEG}^{\rm fin}$ and $\rho_{KEG}^{\rm ideal}$ is upper bounded by $\varepsilon$.}, if $\frac12d_1(\rho_{KEG}^{\rm fin})\le\varepsilon$.
This security criterion is known to satisfy the universal composability \cite{BHLMO05}, and for this reason, considered as the standard.

We note that $d_1(\rho_{KEG}^{\rm fin})$ can also be written as an average with respect to Alice's choice of $g$
\begin{eqnarray}
d_1(\rho^{\rm fin}_{KEG})&=&\sum_g\Pr(G=g)d_1(\rho_{KE}^{{\rm fin},g}),
\label{eq:average_trace_dist_rewritten}
\\
d_1(\rho_{KE}^{{\rm fin},g})&:=&\left\|\rho_{KE}^{{\rm fin},g}-\rho_{KE}^{\rm ideal}\right\|_1.
\label{eq:d_1_defined}
\end{eqnarray}
This form will be convenient in subsequent sections.

\subsection{Existing result on the security (quantum leftover hashing lemma)}

Of various previous results on the security of PA, we focus in this paper on Renner's approach \cite{RennerPhD} using the quantum leftover hashing lemma (LHL, or lemma \ref{lmm:leftover_hashing_lemma} below).
The virtue of Renner's approach is its simplicity;
once one obtains a lower bound on the conditional min-entropy $H_{\min}^{\varepsilon}(\rho_{AE}^{\rm ini}|E)$ of the initial sub-normalized state $\rho_{AE}^{\rm ini}$, the rest of the argument is finished almost automatically, by substituting it to the LHL.
We briefly review this result below.

The conditional min-entropy $H_{\min}(\rho_{AE}|E)$ of a sub-normalized state $\rho_{AE}$ is defined to be the maximum real number $\lambda$, satisfying
\begin{equation}
2^{-\lambda} \II_A\otimes \sigma_E \ge \rho_{AE}
\end{equation}
for a normalized state $\sigma_E$ \cite{RennerPhD,TomamichelPhD}.
Further, the {\it smooth} conditional min-entropy $H_{\min}^{\varepsilon}(\rho_{AE}|E)$ of a sub-normalized state $\rho_{AE}$ is the maximum value of $H_{\min}(\bar{\rho}_{AE}|E)$ of sub-normalized states $\bar{\rho}_{AE}$ that are $\varepsilon$-close (i.e. within distance $\varepsilon$ in terms of the purified distance) to $\rho_{AE}$ ($\bar{\rho}_{AE}\approx_\varepsilon \rho_{AE}$) \cite{TomamichelPhD};
\begin{equation}
H_{\min}^{\varepsilon}(\rho_{AE}|E):=\max_{\substack{{\rm Tr}\bar{\rho}_{AE}\le1,\\ \bar{\rho}_{AE}\approx_\varepsilon \rho_{AE}}}H_{\min}(\bar{\rho}_{AE}|E).
\label{eq:smooth_min_entropy_defined}
\end{equation}
(for details of $\varepsilon$-closeness, see appendix \ref{sec:notation} or ref. \cite{TomamichelPhD}).
A random hash function $G$ is called {\it universal}${}_2$ or {\it two-universal} \cite{CARTER1979143}, if the collision probability of outputs of any distinct input pair $a,a'$ ($a\ne a'$) is bounded as
\begin{eqnarray}
\Pr(G(a)=G(a'))&=&\sum_{g}\Pr(G=g)1[g(a)=g(a')]\nonumber\\
&\le& 2^{-m},
\label{eq:def_universal$_2$_function}
\end{eqnarray}
where function $1[\cdots]$ takes value one if the condition inside brackets hold, and zero otherwise.

With these settings, one can bound $d_1(\rho_{KEG}^{\rm fin})$ of (\ref{eq:average_trace_dist_defined}) by using the following lemma:
\begin{Lmm}[Quantum leftover hashing lemma (LHL) (chapter 5, ref. \cite{RennerPhD})]
\label{lmm:leftover_hashing_lemma}
If random hash function $G$ is universal$_2$, the security of PA scheme $\Pi^{\rm ap}$ is guaranteed by
\begin{eqnarray}
d_1(\rho_{KEG}^{\rm fin})\le 2^{\frac12[m-H_{\min}(\rho_{AE}^{\rm ini}|E)]},
\label{eq:original_leftover_hashing_lemma}
\end{eqnarray}
and in terms of the smooth conditional min-entropy,
\begin{eqnarray}
d_1(\rho_{KEG}^{\rm fin})
\le 2\varepsilon+2^{\frac12[m-H^{\varepsilon}_{\rm min}(\rho_{AE}^{\rm ini}|E)]}.
\label{eq:original_smoothed_leftover_hashing_lemma}
\end{eqnarray}
\end{Lmm}
A proof of lemma \ref{lmm:leftover_hashing_lemma} is given in chapter 5 of Renner's Ph.D. thesis \cite{RennerPhD}.
In that thesis, Renner first bounds the 1-norm, $d_1$ defined in (\ref{eq:d_1_defined}), by the corresponding 2-norm $d_2$.
Then he averages $d_2$ with respect to $G$ and suppresses unwanted terms by exploiting the universal$_2$ property of hash function $G$, so that it is bounded by the right hand side of (\ref{eq:original_leftover_hashing_lemma}).

\section{Leftover hashing lemma as a coding theorem of quantum phase error correction}

In this section, we prove essentially the same result as the LHL, by using a method different from the original paper.
In other words, we prove inequalities which are identical to (\ref{eq:original_leftover_hashing_lemma}) and (\ref{eq:original_smoothed_leftover_hashing_lemma}) except for an inessential constant factor, without relying on the results of ref \cite{RennerPhD}.
In what follows, we will often call our version of inequalities (\ref{eq:original_leftover_hashing_lemma}) and (\ref{eq:original_smoothed_leftover_hashing_lemma}) the {\it LHL-like} bounds.

The proof proceeds in two steps: (i) we present a direct connection between the actual PA algorithm $\Pi^{\rm ap}$ and a certain quantum phase error correction (PEC) algorithm, and (ii) we obtain the LHL-like bounds as a consequence of the coding theorem of the PEC.
More precisely, in step (i), we show that, without loss of security (i.e., without affecting the value of $d_1(\rho_{KEG}^{\rm fin})$), we can modify the actual PA algorithm $\Pi^{\rm ap}$ to another algorithm $\Pi^{\rm vp}$, which has a PEC algorithm embedded inside.
Then in step (ii), we show that the security criterion, $d_1(\rho_{KEG}^{\rm fin})$, can be bounded from above by the failure probability of the PEC algorithm.
As a result of this, the LHL-like bound is obtained as a corollary of a coding theorem of the PEC.

This section can be regarded as a special case of our discussion on quantum key distribution (QKD), given in the next section.
Recall that PA can be interpreted as a special case of QKD where Alice alone generates a final key in the absence of Bob.
%Also recall that, as mentioned in Introduction, there are two approaches to the security proof of QKD, which are generally considered independent: the PEC-based approach and the LHL-based approach.
Then the argument below can be viewed as converting the LHL-based security proof of PA (i.e., the LHL itself) to the PEC-based security proof.

\subsection{Notation}
From now on, we assume that bras and kets indexed by a bit string are all eigenstates of the Pauli operator $Z$; e.g., for $z\in\{0,1\}^n$, $Z^{z'}\ket{z}=(-1)^{z'\cdot z}\ket{z}$.
We also denote Fourier transforms of these $Z$-eigenstates by adding a tilde:
\begin{equation}
\ket{\widetilde{x}}=2^{-n/2}\sum_{z}(-1)^{x\cdot z}\ket{z},\ {\rm for }\ x\in\{0,1\}^n.
\end{equation}
We stress that $\ket{\widetilde{x}}$ is a state $\ket{\widetilde{\cdots}}$ parametrized by a variable $x$, and not $\ket{\cdots}$ parametrized by $\tilde{x}$.
Note that these are eigenstates the Pauli operator $X$; i.e., $X^{x'}\ket{\widetilde{x}}=(-1)^{x'\cdot x}\ket{\widetilde{x}}$.
We will often call $\ket{z}$ the bit basis, and $|\widetilde{x}\rangle$ the phase basis, following the convention of quantum information theory.
For more details of notation, see appendix \ref{sec:notation}.

\subsection{Conversion to a virtual PA scheme}
\label{sec:conversion_to_VPA}
As we have seen in the previous section, the only goal of the security proof of PA is to upper bound $d_1(\rho_{KEG}^\text{fin-ap})$, for  $\rho_{KEG}^\text{fin-ap}$ defined in section \ref{sec:definition_of_PA} and $d_1(\cdots)$ defined in eq. (\ref{eq:average_trace_dist_defined}).
In carrying out this analysis, one loses no generality by considering any scheme,
if its final sub-normalized state equals $\rho_{KEG}^\text{fin-ap}$ of the actual scheme, when one looks only at system ${\cal H}_{KEG}$.
More precisely, we use the following definition:
\begin{Dfn}[Virtual PA schemes]
We say that a PA scheme $(\Pi^{\prime}, \rho^{{\rm ini}\prime})$ is virtual, if the reduced sub-normalized state $\rho^{{\rm fin}\prime}_{KEG}$ for system ${\cal H}_{KEG}$ corresponding to its final sub-normalized state $\rho^{{\rm fin}\prime}=\Pi^{\prime}(\rho^{{\rm ini}\prime})$ equals $\rho^{{\rm fin}\text{-}{\rm ap}}_{KEG}$.
\end{Dfn}
Here, $\rho^{{\rm fin}\text{-}{\rm ap}}_{KEG}$ is the reduced sub-normalized state for ${\cal H}_{KEG}$ corresponding to $\rho^{{\rm fin}\text{-}{\rm ap}}_{KEG}$ of the actual scheme $(\Pi^{\rm ap},\rho^{{\rm ini}\text{-}{\rm ap}})$.
Intuitively, the output of a virtual scheme is exactly the same as that of the actual scheme, if one looks only at system ${\cal H}_{KEG}$.
Hence if one wishes to prove the security of the actual scheme $(\Pi^{\rm ap},\rho^\text{ini-ap})$ by bounding $d_1(\rho_{KEG}^\text{fin-ap})$, it suffices instead to bound $d_1(\rho_{KEG}^{{\rm fin}\prime})$ with $\rho_{KEG}^{{\rm fin}\prime}$ being the output of an arbitrary virtual scheme $(\Pi^{\prime},\rho^{{\rm ini}\prime})$.

The advantage of considering such virtual schemes is that, with an appropriate choice of $(\Pi',\rho')$, one can simplify the mathematical proposition that needs to be proved.
Also, as a result of that, one can even simplify the proof itself.
We stress that virtual schemes are used only as a theoretical tool that is useful in this sense, and one never needs to care how to realize them in practice.

Below we construct an example of virtual schemes by modifying the actual scheme $(\Pi^{\rm ap},\rho^\text{ini-ap})$ in three steps.
Note that indeed none of these changes affects the corresponding final sub-normalized state $\rho_{KEG}^{\rm fin}$ in ${\cal H}_{KEG}$.

\subsubsection{Preparing a purification}
The initial sub-normalized cq state $\rho_{AE}^\text{ini-ap}$ may be replaced with its purification $\ket{\Psi^\text{ini-ap}}_{AEA'}$, with ${\cal H}_{A'}$ being an  ancilla space under Alice's control.

(One example of $\ket{\Psi^\text{ini-ap}}_{AEA'}$ can be defined as follows (cf. eq. (\ref{eq:purification_Phi_defined}) of appendix \ref{sec:proof_of_theorem}): Introduce an ancilla space ${\cal H}_{A'}:={\cal H}_{A_1}\otimes{\cal H}_{A_2}$, with ${\cal H}_{A_1}$ having a sufficiently large dimension, and ${\cal H}_{A_2}$ being an $n$-qubit space.
Then define a purification $\ket{\psi^{a}}_{EA_1}$ of $\rho_E^{{\rm ini},a}$ for each $a\in\{0,1\}^n$, and let $\ket{\Psi^{\rm ini-ap}}_{AEA'}=\sum_{a}\ket{a}_{A}\otimes |\psi^{a}\rangle_{EA_1}\otimes|a\rangle_{A_2}$.)

\subsubsection{PA as the bit basis measurements}
Since hash function $g\in {\cal G}$ is linear and surjective (see the third paragraph of section \ref{sec:definition_of_PA}), the $i$-th bit of $k$ can be represented $k_i=g_i\cdot a$ with a set of linearly independent vectors $g_1,\dots, g_m\in\{0,1\}^n$.
Hence step 2 of $\Pi^{\rm ap}_A$ is equivalent to
measuring eigenvalues $(-1)^{k_i}$ of operators
\begin{equation}
\bar{Z}_i:=Z^{g_i},\ i=1,\dots,m,
\end{equation}
i.e., $\bar{Z}_i\ket{a}=Z^{g_i}\ket{a}=(-1)^{g_i\cdot a}\ket{a}=(-1)^{k_i}\ket{a}$.
We will often call this step the bit basis measurement.

\subsubsection{Quantum error correcting code $PC^g$ defined from hash function $g$}
\label{sec:Code_PCg_defined}
Define a classical error correcting code $C^g$ to be $C^g:=(\ker g)^\perp$, i.e., a linear code having generating matrix $g=(g_1^T,\dots,g_m^T)^T$.
Select a check matrix $h=(h_{1}^T,\dots,h_{n-m}^T)^T$.
By definition, $h$ is an $(n-m)\times n$ matrix satisfying $hg^T=0$.

Then from this classical code, $C^g$, define a quantum error correcting code $PC^g$ by embedding  in phase degrees of freedom of ${\cal H}_A$;
i.e., $PC^g$ is a quantum code characterized by a code space spanned by a basis $\left\{\,\ket{\widetilde{x}}\,|\,x\in C^g\,\right\}$, and syndrome operators
\begin{equation}
\bar{X}_j:=X^{h_j},\ j=1,\dots,n-m.
\end{equation}
Note that $\bar{Z}_i$ and $\bar{X}_j$ commute with each other due to $hg^T=0$.
We note that $PC^g$ is the Hadamard-transform of a classical code having codewords in the bit basis $\left\{\,\ket{x}\,|\,x\in C^g\,\right\}$; thus $PC^g$ corrects only phase errors, not bit errors.
For this reason, we will often call $PC^g$ a phase error correcting code.
We also note that $PC^g$ can alternatively be defined as a Calderbank-Shor-Steane (CSS) code corresponding to a classical code pair $(C_1,C_2)=(\{0,1\}^n, \ker g)$, in the notation of ref. \cite{Nielsen-Chuang}.
In this terminology, $\bar{Z}_i$ and $\bar{X}_{j}$ are 
the logical $Z$ operators and the stabilizers of the CSS code, respectively.

Since $\bar{Z}_i$ and $\bar{X}_j$ commute, syndrome measurements with $\bar{X}_j$ may be inserted before bit basis measurements with $\bar{Z}_i$.
Also for the same reason, we may also insert arbitrary measurements in ${\cal H}_{A'}$, and phase flip operation using $Z_i$, after the syndrome measurement and before the bit basis measurement.

\subsubsection{Virtual PA scheme}
\label{sec:obtained_virtual_PA}
By applying these three changes to the actual PA scheme $(\Pi^{\rm ap}, \rho^\text{ini-ap})$, we obtain virtual PA scheme $(\Pi^{\rm vp},\ket{\Psi^\text{ini-ap}})$, described below.

\begin{oframed}
\noindent\underline{{\bf Virtual PA scheme} $(\Pi^{\rm vp},\ket{\Psi^\text{ini-ap}})$}

\medskip

\noindent{\bf Initial state:}
A purification $\ket{\Psi^\text{ini-ap}}_{AEA'}$ of a sub-normalized cq state $\rho^\text{ini-ap}_{AE}$.

\medskip
\noindent{\bf Algorithm} $\Pi^{\rm vp}_{A|G|A'}$:
\begin{enumerate}
\item {\bf Choice of hash function}:

Choose $g\in{\cal G}$ with probability $\Pr(G=g)$ and announce it publicly by writing it in Eve's space ${\cal H}_G$.
\item {\bf Virtual PA using phase error correcting code $PC^g$} ($\Pi^{{\rm vp},g}_{A|A'}$):
\begin{enumerate}
\item {\bf Phase error correction} $\Pi^{{\rm pec},g}_{A|A'}$:
%using  $PC^g$ assisted by ancilla measurements in ${\cal H}_{A'}$
\begin{enumerate}
\item {\bf Syndrome measurement}:

Obtain syndrome $s=(s_1,\dots,s_{n-m})$ $\in\{0,1\}^{n-m}$ of $PC^g$ by measuring eigenvalues $(-1)^{s_1},\dots$, $(-1)^{s_{n-m}}$ of operators $\bar{X}_{1}$,$\dots$,$\bar{X}_{n-m}$ in ${\cal H}_{A}$, 
\item {\bf Ancilla measurement}:

Determine phase error $e\in\{0,1\}^n$ by measuring ancilla space ${\cal H}_{A'}$ with an operator set $M^{g,s}=\{\, E^{g,s,e}_{A'}\,|\,e\in\{0,1\}^n\}$, satisfying $\sum_{e\in\{0,1\}^n} E^{g,s,e\dagger} E^{g,s,e}=\II_{A'}$.
\item {\bf Phase flip}:

Flip phases by applying $Z^{e}$ in ${\cal H}_A$.
\end{enumerate}
\begin{itemize}
\item[---] We denote the sub-normalized state here by $\rho^\text{pec-vp}$.
\end{itemize}
\item {\bf Bit basis measurement}: Measure eigenvalues $(-1)^{k_1},\dots,(-1)^{k_{m}}$ of operators $\bar{Z}_1,\cdots,\bar{Z}_m$ in ${\cal H}_{A}$, and store hash value $k=(k_1,\dots,k_m)\in\{0,1\}^{m}$ in ${\cal H}_K$.
\end{enumerate}
\end{enumerate}
\begin{itemize}
\item[---] We denote the final sub-normalized state by $\rho^\text{fin-vp}$.
\end{itemize}
\end{oframed}

The following property is evident from the construction.
\begin{Lmm}
\label{lmm:virtuality_virtual_PA}
\begin{equation}
\rho^{{\rm fin}\text{-}{\rm vp}}_{KEG}=\rho^{{\rm fin}\text{-}{\rm ap}}_{KEG},
\end{equation}
and thus the scheme above, $(\Pi^{\rm vp},\ket{\Psi^{{\rm ini}\text{-}{\rm ap}}})$, is indeed virtual.
\end{Lmm}

Recall that the security proof of PA is to show $d_1(\rho^{{\rm fin}\text{-}{\rm ap}}_{KEG})\le\varepsilon$, with $\rho^{{\rm fin}\text{-}{\rm ap}}_{KEG}$ being the final sub-normalized state of the actual scheme $(\Pi^{\rm ap}, \rho^\text{ini-ap})$.
Also note that the above lemma gives $d_1(\rho^{{\rm fin}\text{-}{\rm vp}}_{KEG})=d_1(\rho^{{\rm fin}\text{-}{\rm ap}}_{KEG})$.
Hence, if one wishes to prove the security of the actual scheme, it suffices instead to show $d_1(\rho^{{\rm fin}\text{-}{\rm vp}}_{KEG})\le\varepsilon$ for the output $\rho^{{\rm fin}\text{-}{\rm vp}}_{KEG}$ of the virtual scheme $(\Pi^{\rm vp},\ket{\Psi^\text{ini-ap}})$.

Note that the use of random hash function $g$ in the actual scheme has now been translated to phase error correction $\Pi^{{\rm pec},g}$ using random code $PC^g$ (followed by the $Z$-basis measurement in step 2(b)).
We will see below that the security of the PA scheme $d_1(\rho^{{\rm fin}\text{-}{\rm vp}}_{KEG})$ ($=d_1(\rho^{{\rm fin}\text{-}{\rm ap}}_{KEG})$) is also translated to the performance of the random error correction $\Pi^{{\rm pec},g}$.

Before going on, we give a technical remark.
The phase error correction algorithm $\Pi^{{\rm pec},g}$ is in fact generalized from those typically found in textbooks, such as ref. \cite{Nielsen-Chuang}, in the following two senses.
\begin{itemize}
\item In guessing phase error $e$, Alice exploits a result of an ancilla measurement in ${\cal H}_{A'}$, in addition to syndrome $s$;
i.e., Alice can access an extra hint by measuring ${\cal H}_{A'}$.
\item Operator set $M^{g,s}$ for the ancilla measurement may depend on syndrome $s$ found in the previous step; i.e., Alice has a freedom to optimize her measurement strategy $M^{g,s}$ adaptively after knowing $s$, in order to better determine $e$.
\end{itemize}

\subsection{Reducing the security to the performance of quantum phase error correction}
\label{sec:reducing_security_to_PEC}
Next we illustrate that the security of the virtual PA scheme (and thus also of the actual scheme) can be reduced to the performance of the phase error correction $\Pi^\text{pec}$ embedded inside $\Pi^{\rm vp}$.

From the structure of the virtual scheme, it is evident that the result of $\Pi^\text{pec}$, i.e. $\rho^\text{pec-vp}$, takes the form
\begin{eqnarray}
\rho^\text{pec-vp}_{AA'EG}&=&\sum_{g}\Pr(G=g)\rho^{\text{pec-vp},g}_{AA'E}\otimes\ket{g}\bra{g}_G,\\
\rho^{\text{pec-vp},g}_{AA'E}&:=&\Pi^{{\rm pec},g}_{A|A'}\left(\ket{\Psi^\text{ini-ap}}\bra{\Psi^\text{ini-ap}}_{AA'E}\right).
\label{eq:def_rho_pec_g}
\end{eqnarray}
That is, $\rho^{\text{pec-vp},g}$ is the case of fixed $g$, and $\rho^\text{pec-vp}$ the average with respect to randomly chosen $g$.
Below we discuss these two cases in detail.

\subsubsection{Case of fixed $g$}
Whenever we discuss phase error correction in this paper, we regard $\ket{\widetilde{0}}$ as the {\it correct} phase that should be recovered, and all other phases as errors.
%\begin{equation}
%\ket{\widetilde{0}_L}_A:=2^{-m/2}\sum_{z\in\{0,1\}^m}\ket{z_L}_A=\ket{\widetilde{0}}_A.
%\end{equation}
Accordingly, we evaluate the performance of $\Pi^{\text{pec},g}$ by the phase error probability in ${\cal H}_A$
\begin{eqnarray}
P^{\rm ph}_A
\left(
\rho^{\text{pec-vp},g}
\right)
&=&
%1-{\rm Tr}\left\{\left(\ket{\widetilde{0}}\bra{\widetilde{0}}\right)_A\rho^{{\rm pec},g}_A\right\},
1-\bra{\widetilde{0}}_A\rho^{\text{pec-vp},g}_A\ket{\widetilde{0}}_A,
%\nonumber\\
%&=&
%1-{\rm Tr}
%\left\{
%\left(
%\ket{\widetilde{0}}\bra{\widetilde{0}}_A
%\otimes
%\II_{EB}
%\right)
%\rho^{{\rm pec},g}_{AEB}
%\right\},
\label{eq:failure_prob_PEC_defined}
\end{eqnarray}
with $\rho^{\text{pec-vp},g}_A$ being the reduced sub-normalized state for ${\cal H}_A$ corresponding to $\rho^{\text{pec-vp},g}$; i.e., $\rho^{\text{pec-vp},g}_A={\rm Tr}_{EA'}(\rho^{\text{pec-vp},g}_{AA'E})$.
We will also call this probability the {\it failure probability} of $\Pi^{\text{pec},g}$.

It has been known that $P^{\rm ph}_A\left(\rho^{\text{pec-vp},g}\right)$ can be used to give an upper bound on $d_1(\rho_{KE}^{\text{fin-vp},g})$, i.e., the security measure of hash value $k$ obtained in the subsequent bit basis measurement, against Eve (see, e.g., ref. \cite{HT12}).
For example, consider a simple situation where the phase error probability of the initial $\ket{\Psi^\text{ini-ap}}$ is sufficiently low, and $\Pi^{{\rm pec},g}$ always succeeds: $P^{\rm ph}_A\left(\rho^{\text{pec-vp},g}\right)=0$.
Then Alice's substate is a pure state with a definite phase, $\rho^{\text{pec-vp},g}_{A}=\ket{\widetilde{0}}\bra{\widetilde{0}}_A$, and so Alice's and Eve's joint state must be a tensor product $\rho^{\text{pec-vp},g}_{AE}=\ket{\widetilde{0}}\bra{\widetilde{0}}_A\otimes \rho_E$, which subsequently becomes $\rho_{KE}^{\text{fin-vp},g}=\rho^\text{ideal-vp}_{KE}$ after $\bar{Z}_i$ measurements of step 3;
here $\rho^\text{ideal-vp}_{KE}$ is the ideal state corresponding to $\rho_{KE}^{\text{fin-vp},g}$.
Hence a perfect error correction $\Pi^{{\rm pec},g}$, achieving $P^{\rm ph}_A(\rho^{\text{pec-vp},g})=0$, implies the perfect security $d_1(\rho_{KE}^{\text{fin-vp},g})=0$.
This observation has been extended to the general case including $P^{\rm ph}>0$ as 
\begin{equation}
d_1(\rho_{KE}^{\text{fin-vp},g})\le 2\sqrt2\sqrt{P^{\rm ph}_A\left(\rho^{\text{pec-vp},g}\right)}.
\label{eq:d_1_bounded_by_P_ph_pre}
\end{equation}
Although inequality (\ref{eq:d_1_bounded_by_P_ph_pre}) is a well-known result (see, e.g., ref. \cite{HT12}), we reproduce the proof in appendix \ref{sec:Proof_lemma_phase_error_vs_trace_distance} for readers' convenience.
Note that, due to the virtuality of the virtual PA scheme (lemma \ref{lmm:virtuality_virtual_PA}), inequality (\ref{eq:d_1_bounded_by_P_ph_pre}) also implies a bound for the actual scheme
\begin{equation}
d_1(\rho_{KE}^{\text{fin-ap},g})\le 2\sqrt2\sqrt{P^{\rm ph}_A\left(\rho^{\text{pec-vp},g}\right)}.
\label{eq:d_1_bounded_by_P_ph}
\end{equation}

\subsubsection{Case of randomly chosen $g$}
When hash function $g$ is chosen randomly, we have to consider $d_1(\rho_{KEG}^\text{fin-ap})=\sum_{g}\Pr(G=g)d_1(\rho_{KE}^{\text{fin-ap},g})$, given in eqs. (\ref{eq:average_trace_dist_defined}) and (\ref{eq:average_trace_dist_rewritten}).
By combining relations (\ref{eq:average_trace_dist_rewritten}) and (\ref{eq:d_1_bounded_by_P_ph}) and by applying Jensen's inequality (square root is a concave function), we can bound $d_1(\rho_{KEG}^\text{fin-ap})$ as
\begin{eqnarray}
\lefteqn{d_1(\rho_{KEG}^\text{fin-ap})}\nonumber\\
&=&\sum_{g}\Pr(G=g)d_1(\rho_{KE}^{\text{fin-ap},g})\nonumber\\
&\le& 2\sqrt2\sqrt{\sum_{g}\Pr(G=g) P^{\rm ph}_A\left(\rho^{\text{pec-vp},g}\right)}.
\label{eq:upper_bound_by_average_BLER}
\end{eqnarray}
This means that, in order to bound the secrecy measure $d_1(\rho_{KEG}^\text{fin-ap})$ of the actual PA scheme $(\Pi^{\rm ap},\rho^\text{ini-ap})$, it suffices to bound the average failure probability $\sum_{g}\Pr(G=g)P^{\rm ph}_A\left(\rho^{\text{pec-vp},g}\right)$ of the phase error correction algorithm $\Pi^{{\rm pec},g}$, embedded inside virtual PA algorithm $\Pi^{\rm vp}$.

\subsection{LHL-like bound derived from a coding theorem of phase error correction}
\label{sec:LHL_as_coding_theorem}
If we recall inequality (\ref{eq:original_leftover_hashing_lemma}), the LHL, and compare it with (\ref{eq:upper_bound_by_average_BLER}), we see that both the exponential min-entropy $2^{\frac12[m-H_{\rm min}(\rho_{AE}^{\rm ini}|E)]}$ and the square root failure probability $2\sqrt2 \sqrt{\sum_{g}\Pr(G=g)P^{\rm ph}_A\left(\rho^{\text{pec-vp},g}\right)}$ of error correction $\Pi^{{\rm pec},g}$ are an upper bound on the same quantity, $d_1(\rho_{KEG}^\text{fin-ap})$.
In addition to that, as we have seen in section \ref{sec:conversion_to_VPA}, these two quantities are related by quantum operations: $\rho^{\text{pec-vp},g}$ is obtained by applying $\Pi^{{\rm pec},g}$ to a purification $\ket{\Psi^\text{ini-ap}}$ of $\rho_{AE}^\text{ini-ap}$; see the description of virtual PA $\Pi^{{\rm vp},g}$ and eq. (\ref{eq:def_rho_pec_g}).
Thus it is natural to suspect that these two quantities may actually be equal, or at least comparable to each other.

Theorem \ref{thm_coding_theorem_vpec} below is our answer to this question.
This theorem essentially says that, if we disregard a constant factor $2\sqrt2$, inequality (\ref{eq:upper_bound_by_average_BLER}) is stronger than (\ref{eq:original_leftover_hashing_lemma}), i.e., the quantity on the right hand side of (\ref{eq:upper_bound_by_average_BLER}) can always be bounded from above by that of (\ref{eq:original_leftover_hashing_lemma}).
Interestingly, this relation can be interpreted as a coding theorem for our phase error correction $\Pi^{{\rm pec},g}$
that bounds the average failure probability of $\Pi^{{\rm pec},g}$ in terms of the conditional min-entropy $H_{\rm min}(\rho_{AE}^\text{ini-ap}|E)$ of the initial sub-normalized state $\rho_{AE}^\text{ini-ap}$ :
\begin{Thm}[Coding theorem of the quantum phase error correction $\Pi^{{\rm pec},g}$]
\label{thm_coding_theorem_vpec}
%Let $\ket{\Psi}_{AEA'}$ be a purification of a sub-normalized cq state $\rho_{AE}$ (as in the virtual scheme above).
%Also let random hash function $G$ be universal$_2$.
Under the setting of virtual PA scheme $(\Pi^{\rm vp},\ket{\Psi^{{\rm ini}\text{-}{\rm ap}}})$, let random hash function $G$ be universal$_2$.
Then with an appropriate choice of the ancilla measurement $M^{g,s}$, the average failure probability of $\Pi^{{\rm pec},g}$
% applied on $\ket{\Psi}$
can be bounded as
\begin{eqnarray}
\sum_{g}\Pr(G=g) P^{\rm ph}_A\left(\rho^{{\rm pec}\text{-}{\rm vp},g}\right)\le 2^{m-H_{\rm min}(\rho_{AE}^{{\rm ini}\text{-}{\rm ap}}|E)}.
\label{eq:inequality_coding_theorem}
\end{eqnarray}
%where $\rho^{{\rm pec},g}=\Pi^{{\rm pec},g}(\ket{\Psi}\bra{\Psi})$, as in eq. (\ref{eq:def_rho_pec_g}).
\end{Thm}
We will prove this theorem in appendix \ref{sec:proof_of_theorem}.
To reiterate, the magnitude of the conditional min-entropy $H_{\rm min}(\rho_{AE}^\text{ini-ap}|E)$ does not only guarantee the effectiveness of the leftover hashing lemma, but also that of phase error correction $\Pi^{{\rm pec},g}$, embedded in $\Pi^{\rm vp}$.
It is so effective that if one substitutes inequality (\ref{eq:inequality_coding_theorem}) to (\ref{eq:upper_bound_by_average_BLER}), one can reproduce the LHL for the non-smoothed case.

\begin{Crl}[Quantum LHL-like bound derived from the coding theorem]
\label{lmm:LHL_fom_CT}
Under the setting of virtual PA scheme $(\Pi^{\rm vp},\ket{\Psi^{{\rm ini}\text{-}{\rm ap}}})$, if random hash function $G$ is universal$_2$, we have from (\ref{eq:upper_bound_by_average_BLER}) and (\ref{eq:inequality_coding_theorem}),
\begin{eqnarray}
d_1(\rho_{KEG}^{{\rm fin}\text{-}{\rm ap}})
%&=&\sum_{g}\Pr(G=g)d_1(\rho_{KE}^g)\nonumber\\
%&\le&\sum_{g}\Pr(G=g) 
% 2\sqrt2\sqrt{P^{\rm ph}_A\left(\rho^{{\rm pec},g}_{AEB}\right)}\nonumber\\
&\le&
 2\sqrt2\sqrt{\sum_{g}\Pr(G=g) P^{\rm ph}_A\left(\rho^{{\rm pec}\text{-}{\rm vp},g}\right)}\nonumber\\
&\le& 2\sqrt2\,2^{\frac{1}{2}\left[m-H_{\rm min}(\rho_{AE}^{{\rm ini}\text{-}{\rm ap}}|E)\right]}\nonumber\\
&=& 2^{\frac{1}{2}\left[m-H_{\rm min}(\rho_{AE}^{{\rm ini}\text{-}{\rm ap}}|E)+3\right]}.
\end{eqnarray}
\end{Crl}
This bound is slightly looser than (\ref{eq:original_leftover_hashing_lemma}), by the factor of $2\sqrt2$, but we consider this difference as inessential because in practice it can easily be compensated for by shortening length $m$ of hash value $k$ by three bits.

The discussion above can be considered as an alternative proof of the LHL, obtained by utilizing the PEC-based approach; as PA scheme $\Pi^{\rm ap}$ can be considered as a special case of QKD where Alice alone generates a final key in the absence of Bob.

\subsection{Generalizations}
We present two types of generalization of the above results.

\subsubsection{Case of the smooth conditional min-entropy}
\label{eq:case_of_smooth_min_entropy}
Corollary \ref{lmm:LHL_fom_CT} can be generalized for the smooth conditional min-entropy.

\begin{Crl}[Smoothed quantum LHL-like bound derived from the coding theorem]
\label{crl:smooth_LHL_from_PEC}
Under the setting of virtual PA scheme $(\Pi^{\rm vp},\ket{\Psi^{{\rm ini}\text{-}{\rm ap}}})$, if random hash function $G$ is universal$_2$, we have
\begin{eqnarray}
d_1(\rho_{KEG}^{{\rm fin}\text{-}{\rm ap}})
%\sum_{g}\Pr(G=g)d_1(\rho_{KE}^g)
%&\le& 2\varepsilon+2\sqrt2\,2^{\frac{1}{2}\left(m-H_{\rm min}^\varepsilon(\rho_{AE}|E)\right)}\nonumber\\
&\le& 2\varepsilon+ 2^{\frac{1}{2}\left[m-H_{\rm min}^\varepsilon(\rho_{AE}^{{\rm ini}\text{-}{\rm ap}}|E)+3\right]}.
\end{eqnarray}
\end{Crl}

The proof is essentially the same as in corollary 5.6.1. of ref. \cite{RennerPhD}, but we write it out below to demonstrate that it can be proved entirely within our argument using phase error correction.
Note that in this proof the appropriate strategy $M^{g,s}$ for an ancilla measurement is described explicitly.
\begin{proof}
%In this proof, we abbreviate $\rho^\text{ini-ap}$ and $\ket{\Psi^\text{ini-ap}}$ as $\rho^{\rm ini}$ and $\ket{\Psi^{\rm ini}}$.
By definition of the smooth conditional min-entropy in eq. (\ref{eq:smooth_min_entropy_defined}), there exists a sub-normalized state $\bar{\rho}_{AE}^\text{ini-ap}$ satisfying $\|\bar{\rho}_{AE}^\text{ini-ap}-\rho_{AE}^\text{ini-ap}\|_1\le P(\bar{\rho}_{AE}^\text{ini-ap},\rho_{AE}^{\rm ini})\le\varepsilon$ and $H_{\rm min}^\varepsilon(\rho_{AE}^\text{ini-ap}|E)=H_{\rm min}(\bar{\rho}_{AE}^\text{ini-ap}|E)$.
%It is straightforward to prove that there also exists a purification $\ket{\bar{\Psi}^{\rm ini}}_{AEB}$ of $\bar{\rho}_{AE}^{\rm ini}$, satisfying $P(\ket{\bar{\Psi}^{\rm ini}}, \ket{\Psi^{\rm ini}})=P(\bar{\rho}_{AE}^{\rm ini},\rho_{AE}^{\rm ini})\le \varepsilon$ (cf. Theorem 3.5, \cite{TomamichelPhD}).
We denote the results of applying the actual PA algorithm $\Pi^{{\rm ap},g}$ to $\rho^\text{ini-ap}$ and $\bar{\rho}^\text{ini-ap}$ by $\rho^{\text{fin-ap},g}$ and $\bar{\rho}^{\text{fin-ap},g}$ respectively; i.e., $\rho^{\text{fin-ap},g}=\Pi^{{\rm ap},g}(\rho^\text{ini-ap})$, $\bar{\rho}^{\text{fin-ap},g}=\Pi^{{\rm ap},g}(\bar{\rho}^\text{ini-ap})$.
Also, we denote their ideal sub-normalized states by $\rho_{KE}^{\rm ideal}$ and $\bar{\rho}_{KE}^{\rm ideal}$, respectively (cf. section \ref{sec:security_criterion_PA}).

In this setting, by using the triangle inequality of the $L_1$-distance, $d_1(\rho_{KE}^{\text{fin-ap},g})$ can be bounded as
\begin{eqnarray}
d_1(\rho_{KE}^{\text{fin-ap},g})&=&\left\|\rho^{\text{fin-ap},g}_{KE}-\rho_{KE}^{\rm ideal}\right\|_1
\nonumber\\
&\le&
\left\|\rho^{\text{fin-ap},g}_{KE}-\bar{\rho}^{\text{fin-ap},g}_{KE}\right\|_1\nonumber\\
&&+\left\|\bar{\rho}^{\text{fin-ap},g}_{KE}-\bar{\rho}_{KE}^\text{ideal}\right\|_1+\left\|\bar{\rho}_{KE}^{\rm ideal}-\rho_{KE}^{\rm ideal}\right\|_1\nonumber\\
&=&
\left\|\rho^{\text{fin-ap},g}_{KE}-\bar{\rho}^{\text{fin-ap},g}_{KE}\right\|_1\nonumber\\
&&+d_1\left(\bar{\rho}^{\text{fin-ap},g}_{KE}\right)+\left\|\bar{\rho}_{KE}^{\rm ideal}-\rho_{KE}^{\rm ideal}\right\|_1.
\label{eq:bound_d_1_smooth}
\end{eqnarray}

The first and the third terms on the right hand side of (\ref{eq:bound_d_1_smooth}) can each be bounded by $\varepsilon$:
By definition of the ideal sub-normalized states, we have $\left\|\bar{\rho}_{KE}^{\rm ideal}-\rho_{KE}^{\rm ideal}\right\|_1\le \left\|\rho^{\text{fin-ap},g}_{KE}-\bar{\rho}^{\text{fin-ap},g}_{KE}\right\|_1$.
Since $\Pi^{{\rm ap},g}$ is a CPTP map, the quantities on the right hand side can be further bounded as $\left\|\rho^{\text{fin-ap},g}_{KE}-\bar{\rho}^{\text{fin-ap},g}_{KE}\right\|_1\le \left\|\Pi^{{\rm ap},g}(\rho^\text{ini-ap})-\Pi^{{\rm ap},g}(\bar{\rho}^\text{ini-ap})\right\|_1\le \left\|\rho^\text{ini-ap}-\bar{\rho}^\text{ini-ap}\right\|_1\le\varepsilon$.

Next we bound $d_1\left(\bar{\rho}^{\text{fin-ap},g}_{KE}\right)$ on the right hand side of (\ref{eq:bound_d_1_smooth}).
Choose an arbitrary purification of $\bar{\rho}_{AE}^\text{ini-ap}$ and denote it by $\ket{\bar{\Psi}^\text{ini-ap}}$.
Then from lemma \ref{lmm:virtuality_virtual_PA}, it follows that $d_1\left(\bar{\rho}^{\text{fin-ap},g}_{KE}\right)=d_1\left(\bar{\rho}^{\text{fin-vp},g}_{KE}\right)$ with $\bar{\rho}^{\text{fin-vp},g}=\Pi^{{\rm vp},g}(\ket{\bar{\Psi}^\text{ini-ap}}\bra{\bar{\Psi}^\text{ini-ap}})$.
Also choose the ancilla measurement $M^{g,s}$ in the virtual PA algorithm $\Pi^{{\rm vp},g}$, such that inequality (\ref{eq:inequality_coding_theorem}) of theorem \ref{thm_coding_theorem_vpec} holds for $\ket{\bar{\Psi}^\text{ini-ap}}$ and $\bar{\rho}^\text{ini-ap}$.
Then corollary \ref{lmm:LHL_fom_CT} can be applied, and the average of $d_1\left(\bar{\rho}^{\text{fin-ap},g}_{KE}\right)=d_1\left(\bar{\rho}^{\text{fin-vp},g}_{KE}\right)$ with respect to the random choice of $g$ can be bounded as
\begin{eqnarray}
\lefteqn{\sum_{g}\Pr(G=g)d_1\left(\rho_{KE}^{\text{fin-ap},g}\right)}\\
&\le&2^{\frac{1}{2}\left(m+3-H_{\rm min}(\bar{\rho}_{AE}^\text{ini-ap}|E)\right)}
=2^{\frac{1}{2}\left(m+3-H_{\rm min}^\varepsilon(\rho_{AE}^\text{ini-ap}|E)\right)}.\nonumber
\end{eqnarray}
Then by using eq. (\ref{eq:average_trace_dist_rewritten}), we obtain the lemma.
\end{proof}

\subsubsection{Semi-purification and twirling}
\label{sec:semi_purification_and_twirling}
Thus far we always assumed that the initial sub-normalized state of the virtual PA is a purification $\ket{\Psi^\text{ini-ap}}$ of the initial  sub-normalized cq state $\rho^\text{ini-ap}_{AE}$.
However, as we will see in our discussion on QKD in the next section, it is useful to consider a larger class of inputs that we call {\it semi}-purifications:
\begin{Dfn}[Semi-purification]
\label{dfn:associated_state}
We say that a pure sub-normalized state $\ket{\Phi}_{AEA'}$ is a semi-purification of a cq sub-normalized state $\rho_{AE}$ with respect to ${\cal H}_A$,
if its reduced sub-normalized state for ${\cal H}_{AE}$ (i.e., ${\rm Tr}_{A'}\left(\ket{\Phi}\bra{\Phi}\right)$) equals $\rho_{AE}$ after being measured in the $Z$ basis of ${\cal H}_A$.
That is, $\ket{\Phi}_{AEA'}$ is a semi-purification, if $\rho_{AE}=\Pi^\text{$Z$-mea}_{A}\left({\rm Tr}_{A'}\left\{\ket{\Phi}\bra{\Phi}\right\}\right)=\sum_z\ket{z}\bra{z}_A{\rm Tr}_{A'}\left\{\ket{\Phi}\bra{\Phi}\right\}\ket{z}\bra{z}_A$ (see appendix \ref{sec:notation} for details of the definition of $\Pi^\text{$Z$-mea}$).
\end{Dfn}
It is obvious from this definition that purifications are a special case of semi-purifications.
We note that this larger class of sub-normalized states can always be converted to a purification by the following simple operation, which we call {\it twirling}.
\begin{Dfn}[Twirling]
\label{dfn:twirling}
Twirling operation $\Pi^{\rm tw}_{A|A'}$ consists of two steps: i) add an ancilla $\ket{0}\bra{0}_{A'}$, and then ii) apply CNOT gate $\Pi^{\rm cnot}_{A|A'}=\sum_{z}\ket{z}\bra{z}_{A}\otimes X^{z}_{A'}$ (again, see appendix \ref{sec:notation} for details of the definition of $\Pi^{\rm cnot}_{A|A'}$).
\end{Dfn}
\begin{Lmm}
\label{lmm:twirling}
If $\ket{\Phi}$ is a semi-purification of a sub-normalized cq state $\rho_{AE}$, then $\ket{\Psi}=\Pi^{\rm tw}_{A|A''}(\ket{\Phi})=\Pi^{\rm cnot}_{A|A''}(\ket{\Phi}\otimes\ket{0}_{A''})$ is a purification.
\end{Lmm}
We omit the proof since it is quite straightforward.
By inserting this operation at the beginning of virtual PA $\Pi^{\rm vp}$, we can construct the following scheme.
\begin{oframed}
\noindent\underline{{\bf Virtual PA scheme with twirling} $(\Pi^{\rm vpt},\ket{\Phi^\text{ini-ap}})$}

\medskip
\noindent{\bf Initial state:}

A semi-purification $\ket{\Phi^\text{ini-ap}}_{AA'E}$ of a  sub-normalized cq state $\rho_{AE}^\text{ini-ap}$.

\medskip
\noindent{\bf Algorithm} $\Pi^{\rm vpt}_{A|G|A'|A''}$:
\begin{enumerate}
\item {\bf Twirling}:

Apply $\Pi^\text{\rm tw}_{A|A''}$, described in definition \ref{dfn:twirling}.
\begin{itemize}
\item[---] We denote the result by $\ket{\Psi^\text{tw-vpt}}_{AA'A''E}$.
\end{itemize}
\item {\bf Virtual PA}: Apply $\Pi^{\rm vp}_{A|G|A'A''}$, described in section \ref{sec:obtained_virtual_PA}.
\begin{itemize}
\item[---] We denote the result by $\rho^\text{fin-vpt}$

($=\rho^\text{fin-vpt}_{KAA'A''EG}$).
\end{itemize}
\end{enumerate}
\end{oframed}

The following properties are obvious from the construction:
\begin{Lmm}
\label{lmm:virtuality_virtual_PA_twirling}
$\ket{\Psi^{{\rm tw}\text{-}{\rm vpt}}}_{AA'A''E}$ is a purification of $\rho_{AE}^{{\rm ini}\text{-}{\rm ap}}$. Also,
\begin{equation}
\rho^{{\rm fin}\text{-}{\rm vpt}}_{KEG}=\rho^{{\rm fin}\text{-}{\rm ap}}_{KEG},
\end{equation}
and thus the scheme above, $(\Pi^{\rm vpt},\ket{\Phi^{{\rm ini}\text{-}{\rm ap}}})$, is indeed virtual.
\end{Lmm}

%That is, step 1 of $\Pi^{\rm vpt}$ transforms a semi-purification into a purification.
Hence the situation in step 2 of $(\Pi^{\rm vpt},\ket{\Phi^{{\rm ini}\text{-}{\rm ap}}})$ is equivalent to that of the virtual scheme $(\Pi^{\rm vp}, \ket{\Psi^\text{ini-ap}})$ of section \ref{sec:conversion_to_VPA}.
This means that all the results obtained in this section hold also for $(\Pi^{\rm vpt},\ket{\Phi^{{\rm ini}\text{-}{\rm ap}}})$, including theorem \ref{lmm:LHL_fom_CT} and corollaries \ref{lmm:LHL_fom_CT} and \ref{crl:smooth_LHL_from_PEC}.

\section{Application to the security proof of QKD}
Next we apply our results on privacy amplification (PA), obtained in the previous section, to quantum key distribution (QKD).

As mentioned in section \ref{sec:introduction}, there are two major approaches to the security proof of QKD, which are generally considered independent of each other:
(i) The PEC-based approach, or the Mayers-Shor-Preskill approach \cite{Mayers98,SP00}:
This is essentially the same as our argument using the virtual PA scheme, given in the previous section.
One constructs a virtual QKD protocol which has a phase error correction (PEC) algorithm embedded inside.
The security analysis of the final key is then reduced to bounding the failure probability of the PEC.
(ii) The LHL-based approach, or Renner's approach \cite{RennerPhD}, which uses the leftover hashing lemma (LHL), lemma \ref{lmm:leftover_hashing_lemma} of this paper.

In the previous section, we considered PA scheme $\Pi^{\rm ap}$, which can be thought of as a special case of QKD where Alice alone generates a final key in the absence of Bob.
For this special case, we have shown that there is a direct connection between the two approaches of the security proof;
i.e., we proved essentially the same results as the LHL only by using the PEC-based approach.
Below we show that such connection also exists for QKD.
That is, for QKD protocols in general, we present an explicit procedure for converting a security proof of one approach to that of another approach.
We also show that the obtained security bounds are the same, except for the presence of an inessential constant factor.
This result shows that these two approaches are essentially the same.

The idea is to simply apply our result of the previous section on PA to QKD.
We replace an actual PA algorithm embedded inside a QKD protocol with the virtual PA defined in the previous section, and as a result, obtain a virtual QKD protocol with a virtual PEC embedded inside.
%(This virtual QKD protocol can indeed be considered as a slightly generalized form of those due to Koashi \cite{Koashi}).
%The rest of the argument parallels that of the previous section.
%The trace distance can be bounded by the failure probability $P^{\rm ph}$ of the PEC by using inequality (\ref{eq:upper_bound_by_average_BLER}), and $P^{\rm ph}$ can be again bounded by the minimum entropy $H_{\rm min}(\rho_{AE}|E)$ by using Theorem \ref{thm_coding_theorem_vpec}, the coding theorem of the virtual PEC.

However, there are apparent obstacles in carrying out this idea:
(a) The security criteria of PA and QKD are different; i.e., in addition to the secrecy required by PA, QKD also requires the correctness (see section \ref{sec:standard_security_criteria_QKD}).
(b) In our argument of the virtual PA scheme, we always assumed that the initial sub-normalized state is a purification or a semi-purification of sub-normalized cq state $\rho_{AE}$,   with $H_{\rm min}^\varepsilon(\rho_{AE}|E)$ being bounded from below, but this assumption of semi-purity may not be true for QKD.
For example, each player may generate mixed at any stage of the protocol.

Below we demonstrate that we can indeed overcome these obstacles and complete the conversions between the two approaches, by writing down every step of the conversion explicitly.

\subsection{General QKD scheme}

\subsubsection{Definition}
\label{sec:description_QKD_protocol}

As in the previous sections, we call a set consisting of an initial (sub-normalized) state $\rho^{\rm ini}$ and a QKD protocol $\Pi$ a QKD {\it scheme} $(\Pi, \rho^{\rm ini})$.
In discussions below, we restrict ourselves to actual QKD schemes having the structure $(\Pi^\text{aq},\rho^\text{ini-aq})$ specified below.
To the best of our knowledge, most practical QKD schemes known at the present (e.g., \cite{BB84,B92,TKM03,E91,IWY02,Coles2016,6620246,PhysRevA.73.010302}) can be described as a specific case of this scheme (for detailed argument on the generality of this scheme, see section \ref{sec:assumptions_on_target_protocol}).
\begin{oframed}
\noindent\underline{{\bf Actual QKD scheme} $(\Pi^\text{aq},\rho^\text{ini-aq})$} 

\medskip
\noindent{\bf Initial state}:

A mixed state $\rho^\text{ini-aq}$ ($=\rho_{A'EB'}^\text{ini-aq}$).

\medskip
\noindent{\bf Protocol} $\Pi^\text{aq}$:
\begin{enumerate}
%\item {\bf Initial state generation}:
%Alice, Bob, and Eve together generate an initial state .
\item {\bf Sample measurement}:

Alice and Bob measure ${\cal H}_{A'}$ and ${\cal H}_{B'}$ respectively.
They then discuss their measurement results over a public channel and decide if they abort or continue the protocol.
\item {\bf Sifted key measurement}:

Alice (respectively, Bob) again measures ${\cal H}_{A'}$ (respectively, ${\cal H}_{B'}$) to obtain sifted key $a$ ($b$), stores it in ${\cal H}_{A}$ (${\cal H}_{B}$) in the $Z$-basis.
\begin{itemize}
\item [---] Hence the result is classical in ${\cal H}_A$ and ${\cal H}_B$.
We denote this  sub-normalized state by $\rho^\text{sif-aq}$

 ($=\rho^\text{sif-aq}_{ABA'EB'}$).
\end{itemize}
\item {\bf Key Distillation}:
\begin{enumerate}
\item {\bf Information Reconciliation}:
\begin{enumerate}
\item {\bf Error correction}:

Alice and Bob communicate using encryption.
Bob then corrects bit errors in $b$ and obtain the corrected key $b^{\rm cor}$.
\item {\bf Verification}:

Alice chooses a universal$_2$ function $u:\{0,1\}^n\to\{0,1\}^l$ and announces the choice publicly by writing it in ${\cal H}_E$.
Alice calculates $u(a)$, encrypts it, and sends it to Bob.
Bob then decrypts it and checks if $u(a)=u(b^{\rm cor})$; if not, they return to step i above.
\end{enumerate}
For details on the encryption used here, see section \ref{sec:assumptions_QKD}.
\item {\bf Privacy amplification}:

Alice chooses a universal$_2$ hash function $g:\{0,1\}^n\to\{0,1\}^m$ and announces the choice publicly by writing it in ${\cal H}_G$.
Alice (respectively, Bob) calculates final key $k=g(a)$ (respectively, $k'=g(b^{\rm cor})$) and stores it in ${\cal H}_K$ (${\cal H}_{K'}$).
\end{enumerate}
\begin{itemize}
\item[---] We denote the final sub-normalized state by $\rho^\text{fin-aq}$

(=$\rho^\text{fin-aq}_{KK'GABA'B'E}$).
\end{itemize}
\end{enumerate}
\end{oframed}
The purpose of each step is as follows.
The sample measurement step selects out sub-normalized states that have a certain desirable property, i.e., the types of properties which ensure the final keys extracted later by PA are secure.
In the actual scheme, they do this selection by continuing the protocol only when a subsystem (i.e., samples) of the initial state passes a predetermined test.
The error correction step aims to correct bit errors in Bob's sifted key due to noise in the channel.
Then the verification step verifies that the key is indeed corrected.

\subsubsection{Assumptions}
\label{sec:assumptions_QKD}
We continue to impose the assumptions of section \ref{sec:assumptions_PA} for PA schemes, and also add the following assumptions.

As to Hilbert spaces, ${\cal H}_{A'K}$ is under Alice's control, and ${\cal H}_{BB'K'}$ is under Bob's;
%${\cal H}_{A'}$, ${\cal H}_{B'}$ include degrees of freedom of quantum channels and ancillas used for internal data processing;
${\cal H}_{A}$, ${\cal H}_{B}$ are $n$-qubit spaces, and ${\cal H}_{K}$, ${\cal H}_{K'}$ are $m$-qubit.
As to encryptions in step 3.(a), Alice and Bob use the one-time pad with pre-shared secret keys, e.g., final keys obtained in previous QKD sessions as secret keys.

There are also assumptions on the sample measurement of step 2, but we do not specify them yet;
instead we specify them in sections \ref{sec:LHL_proof_to_PEC_proof} and \ref{sec:conversion_from_PEC_to_LHL_proof}, where we discuss the security proofs.
This is because these assumptions depend strongly on the approaches taken in the security proofs (i.e., the LHL-based and the PEC-based approaches), and thus it is appropriate to specify them at the moment of each proof.

\subsubsection{Remarks on the generality of our actual scheme}
\label{sec:assumptions_on_target_protocol}
%We assume that the initial state $\rho^{\rm ini}$ is generated together by Alice, Bob, and Eve.

We note that most practical QKD protocols known at the present (e.g., \cite{BB84,B92,TKM03,E91,IWY02,Coles2016,6620246,PhysRevA.73.010302}) can be rewritten in the above form without affecting the results of a security proof (i.e., without affecting the security measure $D_1(\rho_{KK'EG}^{\rm fin})$ defined in eq. (\ref{eq:def_D_1_rho_KKEG_QKD})) by the following procedures.
\begin{itemize}
\item If any of the players, Alice, Bob, or Eve, uses a classical random variable $R$, we include it in the initial state $\rho^\text{ini-ap}$, and let him or her measure it later; for example,
\begin{itemize}
\item It is possible that the original protocol uses a one-way quantum communication, instead of an entangled state $\rho_{A'EB'}^\text{ini-ap}$ generated by Eve.
That is, it is possible that the protocol starts with Alice's state $\rho^{{\rm Alice},r}_{A'E}$ occurring with probability $\Pr(R=r)$, and it  then undergoes Eve's attack $\Pi^{\rm Eve}$ and becomes $\rho_{A'EB'}^{r}=\Pi^{\rm Eve}(\rho^{{\rm Alice},r}_{A'E})$.
In such case we rewrite the initial state $\rho_{RA'EB'}^\text{ini-ap}=\sum_{r}\Pr(R=r)\ket{r}\bra{r}_{R}\otimes \rho_{A'EB'}^{r}$, and then let Alice measure $r$ projectively in step 2.
We also redefine space ${\cal H}_{A'R}$ to be ${\cal H}_{A'}$, to simplify notation.
\item Similarly, if Alice performs different measurements depending on $r$, i.e., if she uses different POVMs $M^{r}=\{\,m^{r,i}\,|\,i\,\}$ with probability $\Pr(R=r)$, we replace the initial state with $\rho_{RAA'EB'}=\sum_{r}\Pr(R=r)\ket{r}\bra{r}_R\otimes\rho_{AA'EB'}$, and let her use POVM $M=\{\,\ket{r}\bra{r}_R\otimes m^{r,i}\,|\,r,i\, \}$.
We also redefine space ${\cal H}_{A'R}$ to be ${\cal H}_{A'}$.
\end{itemize}
\item If Alice or Bob publicly announces pieces of information during the protocol in step 1, we assume that they do so by writing them in Eve's space ${\cal H}_E$.
\item We model photon loss as follows. We assume that, in the practical system, detectors can be described mathematically as a measurement discriminating vacuum and non-vacuum states, and a subsequent $Z$-basis measurement. Note that this is indeed true for typical detectors used in practical QKD systems, such as the homodyne measurement, the threshold detector and the photon number resolving detector. In this case we do not loose security by assuming that Eve, instead of Alice, performs the vacuum/non-vacuum discrimination before the sample measurement, since this corresponds to analyzing a case worse than the actual. This situation conforms to our actual QKD scheme $(\Pi^\text{aq},\rho^\text{ini-aq})$.
\item As to post-selection, we assume that it can be modelled mathematically as a measurement that discriminates ``pass'' and ``fail'' sub-normalized states, and a subsequent $Z$-basis measurement. Note that this is certainly true of the many practical protocols including the B92 protocol; see, e.g., \cite{TKM03,Coles2016,6620246,PhysRevA.73.010302}. If we regard the pass/fail discrimination as part of the sample measurement step, the situation conforms to our scheme $(\Pi^\text{aq},\rho^\text{ini-aq})$.
\item As to measurement device-independent (MDI) protocols, we rewrite it such that, at the beginning, Alice and Bob respectively send a half of an entanglement state to Charlie, who then measures them and announces the result;  see, e.g., \cite{Coles2016}. If we regard the result of these operations  the initial state $\rho^\text{ini-aq}$, the situation conforms to our scheme $(\Pi^\text{aq},\rho^\text{ini-aq})$.
\end{itemize}

\subsection{Definition of the security}

\subsubsection{Standard security criterion of QKD ($\varepsilon$-security)}
\label{sec:standard_security_criteria_QKD}
In this subsection, we often abbreviate $\rho^\text{ini-aq}$, $\rho^\text{fin-aq}$ as $\rho^\text{ini}$, $\rho^\text{fin}$.

The goal of the security analysis of QKD is to guarantee that, except with a negligible probability, (i) Alice's and Bob's final keys $k,k'$ in ${\cal H}_{KK'}$ are uniformly distributed over $\{0,1\}^m$, and (ii) $k,k'$ are not accessible by Eve through ${\cal H}_{EG}$.
Hence in defining the security criteria, we may focus on $\rho_{KK'EG}^\text{fin-aq}$, the reduced sub-normalized state for ${\cal H}_{KK'EG}$ corresponding to  $\rho^\text{fin-aq}$ in , i.e., $\rho_{KK'EG}^\text{fin-aq}={\rm Tr}_{AA'BB'}\left(\rho^\text{fin-aq}\right)={\rm Tr}_{AA'BB'}\left(\Pi^\text{aq}\left(\rho^\text{fin-aq}\right)\right)$.
As to the quantitative definition of ``except with a negligible probability'', it is customary to use an approach similar to that used in section \ref{sec:security_criterion_PA} for the case of PA.
That is, one defines the ideal sub-normalized state $\rho_{KK'EG}^{\rm ideal}$ corresponding to $\rho_{KK'EG}^\text{fin-aq}$, and then evaluate the $L_1$-distance between $\rho_{KK'EG}^{\rm ideal}$ and $\rho_{KK'EG}^\text{fin-aq}$.
The details are as follows.

By definition of $\Pi^\text{aq}$, it is evident that $\rho_{KK'EG}^{\rm fin}$ ($=\rho_{KK'EG}^\text{fin-aq}$) generally takes the form
\begin{eqnarray}
\rho_{KK'EG}^{\rm fin}&=&\sum_g\Pr(G=g)\rho_{KK'E}^{{\rm fin},g}\otimes \ket{g}\bra{g}_G,\\
\rho_{KK'E}^{{\rm fin},g}&=&%\Pi^{{\rm qkd},g}(\ket{\Psi}\bra{\Psi})=
\sum_{k,k'}\ket{k}\bra{k}_K\otimes\ket{k'}\bra{k'}_{K'}\otimes \rho^{{\rm fin},g,k,k'}_E.
\label{eq:rho_fin_gkk_defined}
\end{eqnarray}
Corresponding to this sub-normalized state, the ideal sub-normalized state is defined to be
\begin{eqnarray}
\rho_{KK'EG}^{\rm ideal}&:=&\sum_g\Pr(G=g)\rho_{KK'E}^{\rm ideal}\otimes \ket{g}\bra{g}_G,\\
\rho_{KK'E}^{\rm ideal}&:=&\sum_{k}\ket{k}\bra{k}_K\otimes\ket{k}\bra{k}_{K'}\otimes 2^{-m}\rho_E^{\rm fin},\\
\rho_E^{\rm fin}&=&{\rm Tr}_{KK'G}(\rho_{KK'EG}^{\rm fin})\nonumber\\
&=&\sum_g\Pr(G=g)\sum_{k,k'} \rho^{{\rm fin},g,k,k'}_E.
\end{eqnarray}
Note that $\rho_{KK'EG}^{\rm ideal}$ indeed embodies the two conditions for the ideal situation, mentioned in the previous paragraph.

Then we measure the security of QKD protocols by the $L_1$-distance between these two sub-normalized states;
i.e., we use the following definition.
\begin{itemize}
\item $\varepsilon$-security: A QKD protocol is said to be $\varepsilon$-secure, if $\frac12 D_1(\rho_{KK'EG}^{\rm fin})\le\varepsilon$, where
\begin{equation}
D_1(\rho_{KK'EG}^{\rm fin}):=\left\|\rho_{KK'EG}^{\rm fin}-\rho_{KK'EG}^{\rm ideal}\right\|_1.
\label{eq:def_D_1_rho_KKEG_QKD}
\end{equation}
\end{itemize}
As in the case of PA, this security criteria is known to satisfy the universal composability \cite{BHLMO05}, and for this reason, it is considered the standard.

\subsubsection{Separation of the secrecy and the correctness}
Hence the goal of a security proof of QKD is to bound $D_1(\rho_{KK'EG}^{\rm fin})$.
Next we show that this task can further be reduced to proving the secrecy of Alice's final key $k$ alone (corollary \ref{crl:security_reduced_to_Alice}).
The key observation is that the condition $\frac12D_1(\rho_{KK'EG}^{\rm fin})\le\varepsilon$ can be divided into a set of simpler conditions.
Consider the following two conditions \cite{RennerPhD}:
\begin{itemize}
\item $\varepsilon'$-secrecy of Alice's final key alone (disregarding Bob's) against Eve: That is,
$\frac12d_1(\rho_{KEG}^{\rm fin})\le\varepsilon'$, where $d_1(\cdots)$ is the $L_1$-distance defined in section \ref{sec:security_criterion_PA}, and $\rho_{KEG}^{\rm fin}$ is the reduced  sub-normalized state for ${\cal H}_{KEG}$ corresponding to  $\rho^{\rm fin}$.
\item $\varepsilon''$-correctness: Alice's and Bob's final keys $k,k'$ agree except for probability $\varepsilon''$;

i.e., $\Pr\left(K\ne K'\,\land\,\rho^{\rm fin}\right)\le\varepsilon''$, where
\begin{eqnarray}
\lefteqn{\Pr\left(K\ne K'\, \land\, \rho^{\rm fin}\right)}\\
&:=&\sum_{\substack{k,k'\\ k\ne k'}}{\rm Tr}\left\{\left(\ket{k}\bra{k}_{K}\otimes \ket{k'}\bra{k'}_{K'}\right)\rho_{KK'}^{\rm fin}\right\},\nonumber
\end{eqnarray}
with $\rho_{KK'}^{\rm fin}$ being the reduced  sub-normalized state for ${\cal H}_{KK'}$ corresponding to  $\rho^{\rm fin}$.
\end{itemize}
Then $D_1(\rho_{KK'EG}^{\rm fin})$ can be bounded as follows.
\begin{Lmm}[Separation of the secrecy and the correctness]
\label{lmm:security_reduction_to_alice}
\begin{equation}
D_1(\rho_{KK'EG}^{\rm fin})\le d_1(\rho_{KEG}^{\rm fin})+2\Pr\left(K\ne K'\, \land\, \rho^{\rm fin}\right).
\label{eq:separation_secrecy_correctness}
\end{equation}
\end{Lmm}
E.g., the $\varepsilon'$-secrecy and the $\varepsilon''$-correctness imply the $\varepsilon$-security with $\varepsilon=\varepsilon'+\varepsilon''$.
The proof of this lemma is quite straightforward (see appendix \ref{sec:proof_separation_of_trace_distance_and_BEC} of this paper, or section 6.1 of ref. \cite{RennerPhD}).
Furthermore, if we recall that the verification function $u$ of step 3.(a).ii is universal$_2$, the $\varepsilon''$-correctness can be guaranteed automatically.
That is, by using eq. (\ref{eq:def_universal$_2$_function}), the definition of the universal$_2$ property, we obtain the following lemma.
\begin{Lmm}[Bound on the correctness by using universal$_2$ function]
\label{lmm:verification_failure_prob}
\begin{eqnarray}
\lefteqn{\Pr\left(K\ne K'\, \land\, \rho^{\rm fin}\right)}\nonumber\\
&\le& \Pr\left(A\ne B^{\rm cor}\, \land\, U(A)=U(B^{\rm cor})\right)\le 2^{-l}.
\end{eqnarray}
\end{Lmm}
By combining these two lemmas, we finally obtain
\begin{Crl}[Security bounded by the secrecy]
\label{crl:security_reduced_to_Alice}
\begin{equation}
D_1(\rho_{KK'EG}^{\rm fin})\le d_1(\rho_{KEG}^{\rm fin})+2^{-l+1}.
\label{eq:D_1_bounded_by_d1}
\end{equation}
\end{Crl}

\subsection{Conversion to virtual QKD schemes}
Hence the security proof of QKD is now reduced to proving $d_1(\rho_{KEG}^{{\rm fin}\text{-}{\rm aq}})\le\varepsilon'$.
Then we can apply the technique of virtual schemes, which we used for PA in section \ref{sec:conversion_to_VPA}.
For the present case of QKD, we define virtual schemes as follows:
\begin{Dfn}[Virtual QKD schemes]
We say that a QKD scheme $(\Pi',\rho^{{\rm ini}\prime})$ is virtual, if $\rho^{{\rm fin}\prime}_{KEG}$, the reduced sub-normalized state for ${\cal H}_{KEG}$ corresponding to the final  sub-normalized state $\rho^{{\rm fin}\prime}:=\Pi^{\rm aq}(\rho^{{\rm ini}\prime})$, equals that of the actual scheme, $\rho^{{\rm fin}\text{-}{\rm aq}}_{KEG}$.
\end{Dfn}
Thus, as in the case of PA,  if one wishes to prove the security of the actual scheme $(\Pi^{\rm aq},\rho^{{\rm ini}\text{-}{\rm aq}})$, it suffices instead to prove $d_1(\rho_{KEG}^{{\rm fin}\prime})\le\varepsilon'$ for an arbitrary virtual scheme $(\Pi',\rho^{{\rm ini}\prime})$.

Below we modify the actual $(\Pi^\text{aq},\rho^\text{ini-aq})$ step by step, and construct three virtual schemes.
The first virtual scheme has the actual PA embedded inside. 
In section \ref{sec:review_LHL_approach}, we will see that it allows us to prove $d_1(\rho_{KEG}^{{\rm fin}\prime})\le\varepsilon'$ (and thus $d_1(\rho_{KEG}^\text{fin-aq})\le\varepsilon'$) by using the LHL.
On the other hand, the third scheme has virtual PEC embedded inside.
In section \ref{sec:conversion_to_PEC}, we will see that it allows us to prove the security using the PEC.

\subsubsection{Omission of Bob's operations}
Suppose that Bob gives up all his operations after step 2 in $\Pi^\text{aq}$, then the actual scheme $(\Pi^{\rm aq},\rho^\text{ini-aq})$ becomes $(\Pi^{\rm vq1},\rho^\text{ini-aq})$ given below.
\begin{oframed}
\noindent\underline{{\bf Virtual QKD scheme 1} $(\Pi^{\rm vq1},\rho^\text{ini-aq})$}

\medskip
\noindent{\bf Initial state:}

A mixed state $\rho_{A'EB'}^\text{ini-aq}$

\begin{itemize}
\item[---] This is the same as in the actual scheme.
\end{itemize}

\medskip
\noindent{\bf Protocol} $\Pi^{\rm vq1}$:
\begin{enumerate}
\item {\bf Sample measurement}: Same as step 1 of $\Pi^\text{aq}$.

\item {\bf Sifted key measurement}: Alice measures ${\cal H}_{A'}$ to obtain sifted key $a$, and stores it in ${\cal H}_{A}$ in the $Z$-basis.

--- We denote the result by $\rho^\text{sif-vq1}$.
\item {\bf Privacy amplification by Alice alone}:
Alice chooses a universal$_2$ hash function $g:\{0,1\}^n\to\{0,1\}^m$ and announces the choice publicly by writing it in ${\cal H}_G$.
Then she calculates $k=g(a)$ and stores it in ${\cal H}_K$.

--- We denote the result by $\rho^\text{fin-vq1}$.
\end{enumerate}
\end{oframed}
States $\rho^\text{sif-vq1}$ and $\rho^\text{fin-vq1}$, appearing in this scheme, has the following properties:
\begin{Lmm}
\label{lmm:virtualQKD1}
\begin{eqnarray}
\rho^\text{{\rm sif}\text{-}{\rm v1}}_{AE}&=&\rho^{{\rm sif}\text{-}{\rm aq}}_{AE},
\label{eq:rho_sif_vqkd1_eq_rho_sif}
\\
\rho^{{\rm fin}\text{-}{\rm v1}}_{KEG}&=&\rho^\text{{\rm fin}\text{-}{\rm aq}}_{KEG},
\label{eq:rho_fin_vqkd1_eq_rho_fin}
\end{eqnarray}
and thus $(\Pi^{\rm vq1},\rho^{{\rm ini}\text{-}{\rm aq}})$ is indeed virtual.
\end{Lmm}
\begin{proof}
In $(\Pi^{\rm vq1},\rho^\text{ini-aq})$, the initial state $\rho_{A'EB'}^\text{ini-aq}$ and Alice's operations in ${\cal H}_{AKG}$ are identical to those in the actual scheme $(\Pi^\text{aq},\rho^\text{ini-aq})$.
Hence at any stage of the protocol, the reduced  sub-normalized states for ${\cal H}_{AE}$ or for ${\cal H}_{KEG}$ are also identical to those in $(\Pi^\text{aq},\rho^\text{ini-aq})$.
Thus we obtain inequalities (\ref{eq:rho_sif_vqkd1_eq_rho_sif}) and (\ref{eq:rho_fin_vqkd1_eq_rho_fin}).
\end{proof}

\subsubsection{Purification of the overall state}
We modify $(\Pi^{\rm vq1},\rho^\text{ini-aq})$ above further such that the overall  sub-normalized state before step 3, Alice's PA, becomes pure.
The basic idea is to purify the initial state $\rho^\text{ini-aq}$ by introducing new ancilla spaces, and then to rewrite all measurements as a unitary operation followed by a projective measurement; for the details about these techniques, see, e.g., section 2.2.8 of ref. \cite{Nielsen-Chuang}.
More precisely, we make the following changes to $(\Pi^{\rm vq1},\rho^\text{ini-aq})$.
\begin{enumerate}
\item Let Alice execute all of Bob's operations on his behalf, in addition to her own operations.
\item Introduce Alice's new ancilla ${\cal H}_{A''}$, and replace the initial state $\rho^\text{ini-aq}_{A'EB'}$ by its purification $\ket{\Psi^{\rm ini}}_{A'EB'A''}$.
%\item  In the initial state generation step, Eve, together with Alice, generates a purification $\ket{\Psi}_{A'A''EB'}$ of $\rho_{A'EB'}$.
%After that  they do not touch the ancilla space ${\cal H}_{A''}$ until the end of the protocol.
\item Introduce Alice's new ancilla ${\cal H}_{A'''}$, and rewrite step 1, the sample measurement,  as a unitary operation $U^{\rm smp}_{A'B'A'''}$ followed by a  measurement using a projection $P_{A'''}^{\rm smp}$.
--- Note that the decision of `continue' or `abort' is represented by projection $P_{A'''}^{\rm smp}$.
%This can always be done by exploiting techniques elaborated, e.g., in section 2.2.8, ref. \cite{Nielsen-Chuang}.
\item Rewrite step 2, the sifted key measurement, as a unitary operation $U^{\rm sif}_{AA'}$ followed by the $Z$-basis measurement in ${\cal H}_{A}$.
\end{enumerate}
Then we obtain the following scheme.
Here in order to simplify notation, all of Alice's ancilla space ${\cal H}_{A'}\otimes{\cal H}_{A''}\otimes{\cal H}_{A'''}\otimes{\cal H}_{B'}$ are denoted collectively by ${\cal H}_{\bar{A}}$.

\begin{oframed}
\noindent\underline{{\bf Virtual QKD scheme 2} $(\Pi^{\rm vq2},\ket{\Psi^\text{ini-aq}})$}

\medskip
\noindent{\bf Initial state}:

A pure state $\ket{\Psi^\text{ini-aq}}_{\bar{A}E}$.
\begin{itemize}
\item[---] This can also be written as $\ket{\Psi^\text{ini-aq}}_{A'A''B'E}$ and is a purification of $\rho_{A'EB'}^\text{ini-aq}$.
\end{itemize}

\medskip
{\bf Protocol} $\Pi^{\rm vq2}$:
\begin{enumerate}
%\item {\bf Initial state generation}:
%Alice and Eve generate a pure state $\ket{\Psi}_{AA'E}$.
\item  {\bf Sample measurement}: Alice applies unitary operation $U^{\rm smp}_{\bar{A}}$ and then projection $P_{\bar{A}}^{\rm smp}$.

--- We denote the result by $\ket{\Psi^\text{smp-vq2}}$.

\item {\bf Sifted key state preparation}: Alice applies unitary operation $U^{\rm sif}_{A\bar{A}}$.

--- We denote the result by $\ket{\Psi^\text{pre-vq2}}$.
\item {\bf $Z$ basis measurement}: Alice measures ${\cal H}_{A}$ in the $Z$-basis to obtain sifted key $a$.

--- We denote the result by $\rho^\text{sif-vq2}$.
\item {\bf Privacy amplification by Alice alone}: Same as step 3 of $\Pi^{\rm vq1}$.

--- We denote the result by $\rho^\text{fin-vq2}$.
\end{enumerate}
\end{oframed}
This scheme has the following properties.
\begin{Lmm}
\label{lmm:virtualQKD2}
\begin{eqnarray}
\rho^{{\rm sif}\text{-}{\rm vq2}}_{AE}&=&
\rho^{{\rm sif}\text{-}{\rm aq}}_{AE},
\label{eq:rho_sif_vqkd2_eq_rho_sif}
\\
\rho^{{\rm fin}\text{-}{\rm vq2}}_{KEG}&=&\rho^{{\rm sif}\text{-}{\rm aq}}_{KEG},
\label{eq:rho_fin_vqkd2_eq_rho_fin}
\end{eqnarray}
and therefore $(\Pi^{\rm vq2},\ket{\Psi^{{\rm ini}\text{-}{\rm aq}}})$ is indeed virtual.
Also, $\ket{\Psi^{{\rm pre}\text{-}{\rm vq2}}}$ is a semi-purification of $\rho^{{\rm sif}\text{-}{\rm aq}}_{AE}$.
\end{Lmm}
\begin{proof}
It is obvious that changes 1 through 5 above do not affect the reduced sub-normalized state for ${\cal H}_{AKEG}$ at any stage of the protocol.
Hence we have $\rho^\text{sif-vq2}_{AE}=\rho^\text{sif-vq1}_{AE}$ and $\rho^\text{fin-vq2}_{KEG}=\rho^\text{fin-vq1}_{KEG}$.
Then by using lemma \ref{lmm:virtualQKD1}, we obtain inequalities (\ref{eq:rho_sif_vqkd2_eq_rho_sif}) and (\ref{eq:rho_fin_vqkd2_eq_rho_fin}).
After the $Z$-basis measurement (in step 3) is applied to $\ket{\Psi^\text{pre-vq2}}$, it becomes $\rho^{\rm sif-vq2}$, whose reduced sub-normalized state $\rho^{\rm sif-vq2}_{AE}$ equals $\rho_{AE}^\text{sif-aq}$ due to (\ref{eq:rho_sif_vqkd2_eq_rho_sif}).
Therefore $\ket{\Psi^\text{pre-vq2}}$ is a semi-purification of $\rho^{\rm sif}_{AE}$ (cf. definition \ref{sec:semi_purification_and_twirling}).
\end{proof}

\subsubsection{Replacing the actual PA with the virtual PA}
Finally we apply the argument of section \ref{sec:semi_purification_and_twirling}, and replace the actual PA included in $\Pi^{\rm vq2}$ (i.e. steps 3 and 4 of $\Pi^{\rm vq2}$) by the virtual PA algorithm with twirling, $\Pi^{\rm vpt}$.
\begin{oframed}
\underline{{\bf Virtual QKD scheme 3} $(\Pi^{\rm vq3},\ket{\Psi^\text{ini-aq}})$}

\medskip
\noindent{\bf Initial state}: A pure state $\ket{\Psi^\text{ini-aq}}_{\bar{A}E}$
\begin{itemize}
\item[---] This is the same as in virtual QKD scheme 2.
\end{itemize}

\medskip
\noindent{\bf Protocol} $\Pi^{\rm vq3}$:

\begin{enumerate}
\item  {\bf Sample measurement}: Alice applies unitary operation $U^{\rm smp}_{\bar{A}}$ and then projection $P_{\bar{A}}^{\rm smp}$.

\item {\bf Sifted key state preparation}: Alice applies unitary operation $U^{\rm sif}_{A\bar{A}}$.

--- We denote the result by $\ket{\Psi^\text{pre-vq3}}$.
\item {\bf Twirling}: $\Pi^{\rm tw}_{A|T}$, described in definition \ref{dfn:twirling}.

--- We denote the result by $\ket{\Psi^\text{tw-vq3}}$.
\item {\bf Choice of a hash function}:
Alice chooses a universal$_2$ hash function $g:\{0,1\}^n\to\{0,1\}^m$ and announces the choice publicly by writing it in ${\cal H}_G$.
\item {\bf Virtual PA}: $\Pi^{{\rm vp},g}_{A|\bar{A}T}$, described in section \ref{sec:obtained_virtual_PA}; that is,
\begin{enumerate}
\item {\bf Phase error correction using  $PC^g$ assisted by ancilla measurements in ${\cal H}_{\bar{A}T}$} ($\Pi^{{\rm pec},g}_{A|\bar{A}T}$).

--- We denote the result by $\ket{\Psi^\text{pec-vq3}}$.

\item {\bf Bit basis measurement}: Measure eigenvalues $(-1)^{k_1},\dots,(-1)^{k_{m}}$ of operators $\bar{Z}_1,\cdots,\bar{Z}_m$ in ${\cal H}_{A}$, and store hash value $k=(k_1,\dots,k_m)\in\{0,1\}^{m}$ in ${\cal H}_K$.
\end{enumerate}
--- We denote the result by $\rho^\text{fin-vq3}$.
\end{enumerate}
\end{oframed}
We note that steps 1 and 2 of $\Pi^{\rm vq3}$ above are identical to those of $\Pi^{\rm vq2}$;
we wrote them out explicitly for the sake of completeness, since this scheme will be referred to frequently in later subsections.
This scheme has the following properties.
\begin{Lmm}
\label{lmm:virtualQKD3}
\begin{eqnarray}
\rho^{{\rm fin}\text{-}{\rm vq3}}_{KEG}&=&\rho^{{\rm fin}\text{-}{\rm aq}}_{KEG},
\label{eq:rho_fin_vqkd3_eq_rho_fin}
\end{eqnarray}
and thus $(\Pi^{\rm vq3},\ket{\Psi^{{\rm ini}\text{-}{\rm aq}}})$ is virtual.
Also, $\ket{\Psi^{{\rm pre}\text{-}{\rm vq3}}}$ is a semi-purification, and $\ket{\Psi^{{\rm tw}\text{-}{\rm vq3}}}$ is a purification of $\rho^{{\rm sif}\text{-}{\rm aq}}_{AE}$.
\end{Lmm}
\begin{proof}
Eq. (\ref{eq:rho_fin_vqkd3_eq_rho_fin}) follows by repeating the argument of section \ref{sec:semi_purification_and_twirling}.
Since the initial state and steps 1 and 2 are identical in the second and in the third virtual QKD schemes, we have $\ket{\Psi^\text{pre-vq2}}=\ket{\Psi^\text{pre-vq3}}$.
These sub-normalized states are a semi-purification of $\rho^{{\rm sif}\text{-}{\rm aq}}_{AE}$ due to lemma \ref{lmm:virtualQKD2}.
Thus $\ket{\Psi^{{\rm tw}\text{-}{\rm vq3}}}$ is a purification due to lemma \ref{lmm:twirling}.
\end{proof}

\subsection{Converting the LHL-based proofs to the PEC-based proofs}
\label{sec:LHL_proof_to_PEC_proof}

By making use of the virtual schemes obtained above, we next show that there are direct connections between proofs of the LHL-based approach and of the the PEC-based approach.
First we show that the LHL-based proofs can always be converted to those of the PEC-based approach.

\subsubsection{Brief review of the LHL-based security proof}
\label{sec:review_LHL_approach}
We begin by reviewing the LHL-based security proof of QKD briefly \cite{RennerPhD}.

\paragraph{LHL-type QKD scheme}
As mentioned in section \ref{sec:description_QKD_protocol}, the sample measurement step of the actual protocol $\Pi^\text{aq}$ is designed to select out sub-normalized states having a certain desirable property.
When one uses the LHL-based approach for the security proof, it is customary to define the desirable property as follows, such that the LHL (lemma \ref{lmm:leftover_hashing_lemma}) can readily be applied to bound $d_1(\rho_{KEG}^\text{fin-aq})$.
\begin{oframed}
\noindent{\bf Assumption SM-LHL} (Typical assumption on the sample measurement in the LHL-based approach):
\begin{equation}
H_{\rm min}^{\varepsilon'''}(\rho_{AE}^\text{sif-aq}|E)\ge H_{\rm min}^{\rm th}\quad {\rm with}\ H_{\rm min}^{\rm th}\ {\rm being\ a\ constant}.
\label{eq:bound_H_min_actual}
\end{equation}
\end{oframed}
Hence we call QKD schemes with this assumption the {\it LHL-type} QKD schemes.

\paragraph{LHL-based security proof of the LHL-type schemes}
For these schemes, the LHL-based proof usually proceeds as follows.

\begin{itemize}
\item[(i)] The situation of step 3 in $(\Pi^{\rm vq1}, \rho^\text{ini-aq})$ is equivalent to that of the actual PA, defined in section \ref{sec:privacy_amplification}.
Thus one can bound $d_1(\rho_{KEG}^\text{fin-vq1})$ by using the LHL (lemma \ref{lmm:leftover_hashing_lemma}) as
\begin{equation}
d_1(\rho_{KEG}^\text{fin-vq1})\le 2\varepsilon''' + 2^{\frac12[m-H^{\varepsilon'''}_{\rm min}(\rho_{AE}^\text{sif-vq1}|E)]}.
\label{eq:LHL_qkd_v1}
\end{equation}
\item[(ii)] From the virtuality of scheme $(\Pi^{\rm vq1},\rho^\text{ini-aq})$, i.e., from lemma \ref{lmm:virtualQKD1}, we have
\begin{eqnarray}
d_1(\rho_{KEG}^\text{fin-vq1})&=&d_1(\rho_{KEG}^\text{fin-aq}),\\
H^{\varepsilon'''}_{\rm min}(\rho_{AE}^\text{sif-vq1}|E)&=&H^{\varepsilon'''}_{\rm min}(\rho_{AE}^\text{sif-aq}|E).
\end{eqnarray}
Thus inequality (\ref{eq:LHL_qkd_v1}) can be rewritten as
\begin{equation}
d_1(\rho_{KEG}^\text{fin-aq})\le 2\varepsilon'''+2^{\frac12[m-H^{\varepsilon'''}_{\rm min}(\rho_{AE}^\text{sif-aq}|E)]}
\label{eq:LHL_qkd_a}
\end{equation}
\item[(iii)] Combining inequalities (\ref{eq:D_1_bounded_by_d1}), (\ref{eq:bound_H_min_actual}) and (\ref{eq:LHL_qkd_a}), we obtain
\begin{eqnarray}
\lefteqn{D_1(\rho_{KK'EG}^\text{fin-aq})}\nonumber\\
&\le& 2^{-l+1} + 2\varepsilon'''+2^{\frac12[m-H^{\varepsilon'''}_{\rm min}(\rho_{AE}^\text{sif-aq}|E)]}\nonumber\\
&\le&2^{-l+1} + 2\varepsilon'''+ 2^{\frac12[m-H_{\rm min}^{\rm th}]},
\label{eq:LHL_based_security_bound_QKD}
\end{eqnarray}
which completes the security proof.
\end{itemize}

\subsubsection{Conversion to the PEC-based security proofs}
\label{sec:conversion_to_PEC}
Next we show that this type of security proofs can always be converted to those of the PEC-based approach.
That is, for an arbitrary scheme of the LHL-type, with assumption SM-LHL, one can always construct a proof that relies solely on the properties of the PEC, and not on the LHL.

\paragraph{PEC-based proof of the LHL-type schemes}
The proof proceeds by similar steps as in section \ref{sec:review_LHL_approach}, only the argument leading to the bound on $d_1(\rho_{KEG}^\text{fin-aq})$ is quite different.
Namely, (a) We use the third virtual scheme $(\Pi^{\rm vq3}, \rho^\text{ini-vq3})$, not the first, and reduce the security proof to bounding $d_1(\rho_{KEG}^\text{fin-vq3})$.
(b) In order to bound this quantity, we make use of corollary \ref{crl:smooth_LHL_from_PEC}, not of the LHL.
The actual proof proceeds as follows.
\begin{itemize}
\item[(i)] The situation in steps 4 and 5 of $(\Pi^{\rm vq3}, \rho^\text{ini-vq3})$ is identical to that of the virtual PA scheme $(\Pi^{\rm vp},\ket{\Psi^\text{ini-ap}})$ of section \ref{sec:obtained_virtual_PA}, with the input being $\ket{\Psi^\text{ini-ap}}=\ket{\Psi^\text{tw-vq3}}$, a purification of $\rho_{AE}^\text{sif-aq}$ (see lemma \ref{lmm:virtualQKD3}).
Hence we can apply corollary \ref{crl:smooth_LHL_from_PEC} and obtain
\begin{equation}
d_1(\rho_{KEG}^\text{fin-vq3})\le 2\varepsilon''' + 2^{\frac12[m-H^{\varepsilon'''}_{\rm min}(\rho_{AE}^\text{sif-aq}|E)+3]}.
\label{eq:LHL_qkd_v3}
\end{equation}
\item[(ii)] From the virtuality of scheme $(\Pi^{\rm vq3}, \rho^\text{ini-vq3})$, i.e., from lemma \ref{lmm:virtualQKD3}, we have
\begin{equation}
d_1(\rho_{KEG}^\text{fin-vq3})=d_1(\rho_{KEG}^\text{fin-aq}),
\end{equation}
and thus inequality (\ref{eq:LHL_qkd_v3}) can be rewritten as
\begin{equation}
d_1(\rho_{KEG}^\text{fin-aq})\le 2\varepsilon'''+2^{\frac12[m-H^{\varepsilon'''}_{\rm min}(\rho_{AE}^\text{sif-aq}|E)+3]}
\label{eq:bound_d1_qkd_PEC}
\end{equation}
\item[(iii)] Combining inequalities (\ref{eq:D_1_bounded_by_d1}), (\ref{eq:bound_H_min_actual}) and (\ref{eq:bound_d1_qkd_PEC}), we obtain
\begin{eqnarray}
\lefteqn{D_1(\rho_{KK'EG}^\text{fin-aq})}\nonumber\\
&\le& 2^{-l+1} + 2\varepsilon'''+2^{\frac12[m-H^{\varepsilon'''}_{\rm min}(\rho_{AE}^\text{sif-aq}|E)+3]}\nonumber\\
&\le&2^{-l+1} + 2\varepsilon'''+ 2^{\frac12[m-H_{\rm min}^{\rm th}+3]}.
\label{eq:bound_on_D1_by_PEC_for_LHL_type}
\end{eqnarray}
which completes the security proof.
\end{itemize}
We again stress that the essential part of this proof, inequality (\ref{eq:LHL_qkd_v3}), is obtained solely from the properties of the PEC algorithm $\Pi^{\rm pec}$, embedded in the virtual scheme $(\Pi^{\rm vq3}, \rho^\text{ini-vq3})$.
In this sense, we say that this is a proof by the PEC-based approach.
In section \ref{sec:brief_review_PEC_approach}, we also give a direct comparison between this and the conventional forms of the PEC-based proofs.

Inequality (\ref{eq:bound_on_D1_by_PEC_for_LHL_type}) is slightly weaker than the counterpart (\ref{eq:LHL_based_security_bound_QKD}) obtained in the LHL-approach by a constant factor of 3/2 appearing in an exponent of the r.h.s. (cf. corollary \ref{crl:smooth_LHL_from_PEC} for the case of PA).
We regard this factor as inessential because it can easily be compensated for in practice, e.g., by setting length $m$ of final key $k$ to be 3 bits shorter.

\subsection{Converting the PEC-based proofs to the LHL-based proof}
\label{sec:conversion_from_PEC_to_LHL_proof}
Next we present the conversion of the reverse direction.
That is, we show that any security proofs of the PEC-based approach can be converted to those of the LHL-based approach.

\subsubsection{Brief review of the PEC-based security proof}
\label{sec:brief_review_PEC_approach}

We again begin by reviewing the existing method (see, e.g., \cite{Koashi, KoashiNJP09,H07,HT12}).

\paragraph{PEC-type QKD schemes}
Again, as in the case of the LHL-approach, the sample measurement step of the actual protocol $\Pi^\text{aq}$ is supposed to select out sub-normalized states having a certain desirable property.
However, when using the PEC-based approach \cite{Koashi, KoashiNJP09,H07,HT12}, it is customary to define the desirable property with respect to the pure sub-normalized state $\ket{\Psi^\text{pre-vq3}}$ appearing in the third virtual scheme.
Note that this is in strong contrast to the case of the LHL-based approach, where the property was defined with respect to $\rho^{\rm sif}$ in the actual scheme.
In addition, it is also customary to define the desirable property of $\ket{\Psi^\text{pre-vq3}}$ in terms of its phase probability distribution $p_X$, defined below.
\begin{Dfn}
The distribution of phase degrees of freedom $X^{{\rm pre}\text{-}{\rm vq3}}$ of $\ket{\Psi^{{\rm pre}\text{-}{\rm vq3}}}$ is defined by
\begin{equation}
\Pr(X^{{\rm pre}\text{-}{\rm vq3}}=x):=\bra{\widetilde{x}}_{A}\rho^{{\rm pre}\text{-}{\rm vq3}}_A\ket{\widetilde{x}}_A,
\label{eq:Pr_X_prevq3_defined}
\end{equation}
where $\rho^{{\rm pre}\text{-}{\rm vq3}}_A$ is the reduced sub-normalized state for ${\cal H}_A$ corresponding to $\ket{\Psi^{{\rm pre}\text{-}{\rm vq3}}}$.
%$\rho^{{\rm pre}\text{-}{\rm vq3}}={\rm Tr}_{\bar{A}E}\left\{\ket{\Psi^{{\rm pre}\text{-}{\rm vq3}}}\bra{\Psi^{{\rm pre}\text{-}{\rm vq3}}}_{A\bar{A}E}\right\}$.
\end{Dfn}

In this setting, one assumes that the randomness of $X^{{\rm pre}\text{-}{\rm vq3}}$ is sufficiently small, such that the PEC $\Pi^{\rm pec}$ inside the virtual protocol $\Pi^{\rm vq3}$ almost always succeeds.

\begin{oframed}
\noindent{\bf Assumption SM-PEC} (Typical assumptions on the sample measurement in the PEC-based approach):

\medskip
The randomness of $X^{{\rm pre}\text{-}{\rm vq3}}$ is sufficiently small; e.g.,
\begin{itemize}
\item Example 1: The error probability $\Pr(X^{{\rm pre}\text{-}{\rm vq3}}\ne 0)$ is bounded from above.
\item Example 2: The  Shannon entropy $H(X^{{\rm pre}\text{-}{\rm vq3}})$ is bounded from above.
\end{itemize}
\end{oframed}
We will call QKD schemes with these types of assumptions the {\it PEC-type} QKD schemes.

\paragraph{PEC-based proof of the PEC-type schemes}
The basic idea is the same as in the proof of the LHL-type schemes given in section \ref{sec:conversion_to_PEC}.
The only difference is that, in order to bound the failure probability $P^{\rm ph}_A$ of the PEC, we may use any types of a coding theorems of the PEC, besides theorem \ref{thm_coding_theorem_vpec}.
The actual proof proceeds as follows.
\begin{itemize}
\item[(i)] By using a certain coding theorem for the PEC (e.g., theorem \ref{thm_coding_theorem_vpec}, but not necessarily restricted to it), one bounds the average failure probability of the PEC $\Pi^{{\rm pec},g}$, embedded inside the third virtual scheme  $(\Pi^{\rm vq3}, \rho^\text{ini-vq3})$ as
\begin{equation}
P^{\rm ph}_A(\rho^\text{pec-vq3})\le P^{\rm th},
\label{eq:P_Ph_bounded_by_P_th}
\end{equation}
where function $P^{\rm ph}(\cdots)$ is defined in (\ref{eq:failure_prob_PEC_defined}), and $P^{\rm th}$ is a constant determined by  assumption SM-PEC.
Then by substituting (\ref{eq:P_Ph_bounded_by_P_th}) in inequality (\ref{eq:upper_bound_by_average_BLER}), one obtains the bound
\begin{equation}
d_1(\rho_{KEG}^\text{fin-vq3})\le 2\sqrt2 \sqrt{P^{\rm th}}.
\label{eq:PEC_proof_of_PEC1}
\end{equation}
\item[(ii)] From the virtuality of the third virtual scheme (lemma \ref{lmm:virtualQKD3}), we have $d_1(\rho_{KEG}^\text{fin-vq3})=d_1(\rho_{KEG}^\text{fin-aq})$.
Thus inequality (\ref{eq:PEC_proof_of_PEC1}) can be rewritten as
\begin{equation}
d_1(\rho_{KEG}^\text{fin-aq})\le 2\sqrt2 \sqrt{P^{\rm th}}.
\label{eq:PEC_proof_of_PEC2}
\end{equation}
\item[(iii)] Combining inequalities (\ref{eq:D_1_bounded_by_d1}) and (\ref{eq:PEC_proof_of_PEC2}), we obtain
\begin{equation}
D_1(\rho_{KK'EG}^\text{fin-aq})\le 2^{-l+1} + 2\sqrt2 \sqrt{P^{\rm th}},
\end{equation}
which completes the security proof.
\end{itemize}

Note that the PEC-based proof for the LHL-type schemes, given in section \ref{sec:conversion_to_PEC}, was indeed a special case of this argument:
There, the randomness of $X^{{\rm pre}\text{-}{\rm vq3}}$ was measured using the conditional min-entropy $H_{\rm min}^{\varepsilon'''}(\rho_{AE}^\text{sif-aq}|E)$, and theorem \ref{thm_coding_theorem_vpec} was used as an example of inequality (\ref{eq:P_Ph_bounded_by_P_th}).

\subsubsection{Special case using R\'{e}nyi entropy of degree one half}
The description of the PEC-type schemes above becomes drastically simple by using the following quantity as the measure of the randomness of $X^{{\rm pre}\text{-}{\rm vq3}}$:
\begin{Dfn}
The R\'{e}nyi  entropy of degree $1/2$ of a probability distribution $\Pr(X=x)$ is defined by
\begin{equation}
H_{1/2}(X)=2\log\left(\sum_{x}\Pr(X=x)^{1/2}\right)^2.
\end{equation}
\end{Dfn}
More precisely, suppose one uses the particular form of assumption SM-PEC, which takes the from
\begin{oframed}
\noindent{\bf Assumption SM-PEC-R\'{e}nyi} (An example of assumption SM-PEC using the R\'{e}nyi entropy):
\begin{equation}
H_{1/2}(X^{{\rm pre}\text{-}{\rm vq3}})\le H_{1/2}^{\rm th}\quad {\rm with}\ H_{1/2}^{\rm th}\ {\rm being\ a\ constant}.
\label{eq:bound_H_12X}
\end{equation}
\end{oframed}
Then this condition becomes reduced to a special case of assumption SM-LHL, with $H_{\rm min}^{\rm th}=n-H_{1/2}^{\rm th}$ and $\varepsilon'''=0$.
\begin{Lmm}
\label{lmm:duality}
Assumption SM-PEC-R\'{e}nyi implies
\begin{equation}
H_{\rm min}(\rho^{{\rm sif}\text{-}{\rm aq}}_{AE}|E)\ge n-H_{1/2}^{\rm th}.
\label{eq:H_min_bounded_by_H_12_th}
\end{equation}
\end{Lmm}
The proof is a direct consequence of the duality relation between the conditional min- and max- entropies \cite{TomamichelPhD}; see appendix \ref{sec:proof_duality}.

Thus the situation is now reduced to that of the LHL-type schemes of section \ref{sec:LHL_proof_to_PEC_proof}, and the security can be proved by the LHL-based approach and also by the PEC-based approach.
The proof of the PEC-based approach can be given by repeating the discussion of section \ref{sec:conversion_to_PEC} or \ref{sec:brief_review_PEC_approach}, and the result is the security bound
\begin{equation}
D_1(\rho_{KK'EG}^\text{fin-aq})\le2^{-l+1} + 2^{\frac12[m-n+H_{\rm 1/2}^{\rm th}+3]}.
\end{equation}
If one wishes to convert this proof to the LHL-based proof, one can use the discussion of section \ref{sec:review_LHL_approach} and obtain
\begin{equation}
D_1(\rho_{KK'EG}^\text{fin-aq})\le2^{-l+1} + 2^{\frac12[m-n+H_{\rm 1/2}^{\rm th}]}.
\end{equation}

\section{Conclusion}
We have shown the equivalence between the two major approaches to the security proof of quantum key distribution (QKD), which had generally been regarded as independent of each other; namely, the phase error correction (PEC)-based approach and the leftover hashing lemma (LHL)-based approach.
These are also referred to as the Shor-Preskill-Mayers approach and as Renner's approach.

In order to show the equivalence, we gave an explicit procedure for converting security proofs by one approach to those of another approach.
The conversions are made step by step in a constructive manner so that the relation between the two approaches is kept transparent.
The security bounds obtained are the same regardless of the approach, except for the presence of an inessential constant factor.

\ 

\noindent{\bf Acknowledgment}

We thank Kiyoshi Tamaki and Akihiro Mizutani for comments and encouragement.
This work was partially supported by ImPACT Program of Council for Science, Technology and Innovation (Cabinet Office, Government of Japan).

\appendices

\section{Notation}
\label{sec:notation}
Random variables are written in capital letters.
Modulo 2 is always assumed in arithmetic of binary strings; e.g., the inner product of binary strings $x=(x_1,\cdots,x_d)$, $z=(z_1,\cdots,z_d)\in\{0,1\}^d$ is denoted $x\cdot z=\sum_{i=1}^d x_iz_i\, \mod\,  2$.
We use the usual representation of the Pauli matrices, $Z=\begin{pmatrix}1&0\\0&-1 \end{pmatrix}$, $X=\begin{pmatrix}0&1\\1&0 \end{pmatrix}$.
In a $d$-qubit space, $Z,X$ of the $i$-th qubit are often denoted $Z_i$, $X_i$.
A tensor product of $Z_i$ is often abbreviated as $Z^z=\bigotimes_{i-1}^d Z_i^{z_i}=Z^{z_1}\otimes\cdots\otimes Z^{z_d}$; $X^x$ is also defined similarly.
Eigenvectors of $Z$ operators are denoted $\ket{z}$ for $z\in\{0,1\}^d$, satisfying $Z^{z'}\ket{z}=(-1)^{z'\cdot z}\ket{z}$.
Eigenvectors of $X$ operators are chosen to be their Fourier (or Hadamard) transforms
\begin{equation}
\ket{\widetilde{x}}=2^{-d/2}\sum_{z}(-1)^{x\cdot z}\ket{z}
\end{equation}
for $x\in\{0,1\}^d$; note that these vectors satisfy $X^{x'}\ket{\widetilde{x}}=(-1)^{x'\cdot x}\ket{\widetilde{x}}$.
We often call $\ket{z}$ the bit basis, and $|\widetilde{x}\rangle$ the phase basis.

Completely positive trace-preserving (CPTP) maps are denoted by $\Pi$; e.g., $\Pi^{\rm vp}_{A|G|A'}$ for the virtual PA algorithm defined in section \ref{sec:obtained_virtual_PA}.
Subscripts of $\Pi$ indicate systems on which it operate. 
When we wish to indicate clearly that different systems have a different role (such as the input and the output registers, and ancillas), we insert symbol `$|$' between them.

We denote by operator $\Pi^\text{$Z$-mea}$ an operation of erasing non-diagonal elements in the $Z$ basis, $\Pi^\text{$Z$-mea}_D(\rho_D)=\sum_{z\in\{0,1\}^d}\ket{z}\bra{z}_D\rho_D\ket{z}\bra{z}_D$, and call it the {\it measurement} in the $Z$-basis, with a slight abuse of terminology.
Similarly, we also define the $X$-basis measurement by $\Pi^\text{$X$-mea}_D(\rho_D)=\sum_{x\in\{0,1\}^d}\ket{\widetilde{x}}\bra{\widetilde{x}}_D\rho_D\ket{\widetilde{x}}\bra{\widetilde{x}}_D$.
Whenever we discuss phase error correction, we choose $\ket{\widetilde{0}}$ to be the correct phase, and consider all other phases as errors; hence the phase error rate of a sub-normalized state $\rho_{D}$ is given by $P^{\rm ph}_{D}(\rho_{D}):=1-{\rm Tr}\left\{\ket{\widetilde{0}}\bra{\widetilde{0}}_{D}\rho_{D}\right\}$.

In a composite space of two $d$-qubit spaces ${\cal H}_{A} \otimes {\cal H}_{B}$, we define a CNOT operation $N_{A|B}$ by
\begin{eqnarray}
N_{A|B}
&=&\sum_{z}\Ket{z}\Bra{z}_{A}\otimes X^{z}_{B}=\sum_{x} Z^{x}_{A}\otimes\Ket{\widetilde{x}}\Bra{\widetilde{x}}_{B}.
\label{eq:N_D1_D2_defined}
\end{eqnarray}

A positive semi-definite operator $\rho$ is normalized if ${\rm Tr}\rho=1$, and sub-normalized if ${\rm Tr}\rho\le1$.
A (sub-normalized) state $\rho_{AB}$ in ${\cal H}_{A}\otimes{\cal H}_{B}$ is a (sub-normalized) classical-quantum (cq) state, if its ${\cal H}_{A}$ part is diagonalized with a basis $\{\, \ket{v_i}\,\}$: $\rho_{AB}=\sum_i\ket{v_i}\bra{v_i}_A\otimes\rho_B^i$.

The $L_1$-norm (or 1-norm) of a matrix $M$ is defined by $\left\|M\right\|_1:={\rm Tr}\left\{\left|M\right|\right\}={\rm Tr}\left\{\sqrt{M^\dagger M}\right\}$.

For sub-normalized states $\rho,\tau$, the $L_1$-distance is defined to be $\|\rho-\tau\|_1$, and the generalized quantum fidelity $F$ and the purified distance $P$ are defined to be
\begin{eqnarray}
F(\rho, \tau)&:=&\left\|\sqrt{\rho}\sqrt{\tau}\right\|_1+\sqrt{(1-{\rm Tr}\rho)(1-{\rm Tr}\tau)}\nonumber\\
&=&
{\rm Tr}\left\{\sqrt{\sqrt{\rho}\tau\sqrt{\rho}}\right\}\nonumber\\
&&+\sqrt{(1-{\rm Tr}\rho)(1-{\rm Tr}\tau)},\\
P(\rho,\tau)&:=&\sqrt{1-F(\rho,\tau)^2}
\end{eqnarray}
(see refs. \cite{PhysRevA.66.042304,Rastegin_2003,2006quant.ph..2112R,PhysRevA.71.062310,TomamichelPhD}).
$F$ equals the usual fidelity, as defined in \cite{Nielsen-Chuang}, when at least one of $\rho,\tau$ is normalized.
$P$ is an upper bound on the $L_1$-distance: $\|\rho-\tau\|_1\le P(\rho,\tau)$.
Sub-normalized states $\rho,\tau$ are called $\varepsilon$-close, and denoted $\rho\approx_\varepsilon \tau$, when $P(\rho,\tau)\le\varepsilon$.

\section{Proof of theorem \ref{thm_coding_theorem_vpec}}
\label{sec:proof_of_theorem}
First note that the both sides of inequality (\ref{eq:inequality_coding_theorem}) scale linearly with respect to the norm of input sub-normalized state, ${\rm Tr}(\rho^\text{ini-ap})={\rm Tr}(\rho^{\text{pec-vp},g})$.
Hence it suffices to prove this theorem for the normalized case, i.e. for ${\rm Tr}(\rho^\text{ini-ap})=1$.
For this particular case, we divide theorem \ref{thm_coding_theorem_vpec} into two lemmas as follows.
\begin{Lmm}
\label{lmm:Pph_upperbound}
Consider the setting of virtual PA scheme $(\Pi^{\rm vp},\ket{\Psi^{{\rm ini}\text{-}{\rm ap}}})$, and suppose that $\rho^{{\rm ini}\text{-}{\rm ap}}$ is normalized.
Then by an appropriate choice of the ancilla measurement $M^{g,s}$, the failure probability $P^{\rm ph}_A\left(\rho^{{\rm pec}\text{-}{\rm vp},g}\right)$ of phase error correction $\Pi^{{\rm pec},g}$ can be bounded from above as
\begin{eqnarray}
P^{\rm ph}_A\left(\rho^{{\rm pec}\text{-}{\rm vp},g}\right)
 &\le&1-F\left(\rho_{KE}^{\rm ideal}, \rho_{KE}^{{\rm fin}\text{-}{\rm vp},g}\right)^2,
\label{eq:lemma1}
\end{eqnarray}
where $\rho_{KE}^{\rm ideal}$ is defined in (\ref{eq:rho_KE_ideal_defined}).
\end{Lmm}
%Note that the right hand side of (\ref{eq:lemma1}) coincides with the squared purified distance $P(\rho_{KE}^{\rm ideal}, \rho_{KE}^g)^2$, introduced in ref. \cite{TomamichelPhD}.

\begin{Lmm}
\label{lmm:universal_hashing}
Under the setting of lemma \ref{lmm:Pph_upperbound}, by averaging the right hand side of Inequality (\ref{eq:lemma1}) with respect to a universal$_2$ hash function $G$, we obtain
\begin{align}
&\sum_{g}\Pr(G=g)\,\left( 1-F\left(\rho_{KE}^{\rm ideal}, \rho_{KE}^{{\rm fin}\text{-}{\rm vp},g}\right)^2\right)\nonumber\\
&\quad\quad\le2^{m-H_{\rm min}(\rho_{AE}^{{\rm ini}\text{-}{\rm ap}}|E)}.
\label{eq:lemma2}
\end{align}
\end{Lmm}

Below we give the proofs of these lemmas.

\subsection{Proof of lemma \ref{lmm:Pph_upperbound}}
\label{sec:proof_P_ph_upperbound}

In this proof, we abbreviate $\rho^\text{ini-ap}$, $\ket{\Psi^\text{ini-ap}}$ and $\rho^{{\rm pec}\text{-}{\rm vp},g}$, as $\rho^\text{ini}$, $\ket{\Psi^\text{ini}}$ and $\rho^{{\rm pec},g}$, respectively.

\subsubsection{Reduction to a special case}
\label{sec:reduction_to_special_case}
First we show that it suffices to prove this lemma only for the special case where $g$ is chosen such that hash value $k$  equals the most significant $m$ bits of sifted key $a$, and the syndrome measurement are phase value of the least significant $n-m$ qubits.

\begin{proposition}
There exists a unitary transform $V$ in an $n$-qubit space that transforms the logical $Z$ operators $\bar{Z}_i=Z^{g_i}$ and the syndrome measurement operators $\bar{X}_i=X^{h_i}$ to the normal Pauli matrices acting on the $i$th qubit:
\begin{eqnarray}
V \bar{Z}_i V^\dagger&=&Z_i,\label{eq:barZ_transformed_by_V}\\
V \bar{X}_i V^\dagger&=&X_i.\label{eq:barX_transformed_by_V}
\end{eqnarray}
\end{proposition}
\begin{proof}
Recall that $\{g_1,\cdots,g_m\}$ and $\{h_1,\cdots,h_{n-m}\}$, given in section \ref{sec:Code_PCg_defined}, are orthogonal to each other, and choose vectors $g_{m+1},\dots,g_n$ satisfying $g_{i+m}\cdot h_j=\delta_{ij}$ for $1\le i,j\le n-m$.
Consider a permutation over $n$-bit strings defined by a linear transform $a'= av$ with $v:=(g_1^T,\dots,g_n^T)$, and let $V$  be the corresponding unitary transform in the $n$-qubit space
\begin{eqnarray}
V\ket{a}&=&\ket{av}.
%=\Ket{(g_1\cdot a,\dots,g_n\cdot a)}_A=\Ket{(zg_1^T,\dots,zg_n^T)}_A
\label{eq:U_A_defined}
\end{eqnarray}
Then relations (\ref{eq:barZ_transformed_by_V}), (\ref{eq:barX_transformed_by_V}) follow from a straightforward calculation.
%\begin{eqnarray}
%V\ket{\widetilde{x}}&=&
%V\sum_z(-1)^{x\cdot z}\ket{z}=\sum_z(-1)^{x(v^{-1})^Tv^T z^T}\ket{zv}=\sum_{z}(-1)^{x(v^{-1})^T (zv)^T}\ket{zv}\nonumber\\
%&=&\sum_{z}(-1)^{(x(v^{-1})^T)\cdot (zv)}\ket{zv}=\ket{\widetilde{x(v^{-1})^T}},\\
%VZ^{z} V^\dagger&=&
%V\sum_{z'}(-1)^{z\cdot z'}\ket{z'}\bra{z'} V^\dagger=\sum_{z'}(-1)^{z\cdot z'}\Ket{z'v}\Bra{z'v}\nonumber\\
%&=&\sum_{z'}(-1)^{z\cdot (z'v^{-1})}\Ket{z'}\Bra{z'}=\sum_{z'}(-1)^{z(v^{-1})^Tz'^T}\Ket{z'}\Bra{z'}\nonumber\\
%&=&\sum_{z'}(-1)^{z(v^{-1})^T\cdot z'}\Ket{z'}\Bra{z'}=
%Z^{z(v^{-1})^T},\\
%VX^{x} V^\dagger&=&
%X^{x((v^{-1})^T)^{-1})^T}=
%X^{xv},
%\end{eqnarray}
%where we set $z''=z'v$ in the second equation, and used the fact that $1=(AA^{-1})^T=(A^{-1})^TA^T$ and thus $(A^T)^{-1}=(A^{-1})^T$ holds for any invertible matrix $A$ in the third equation.
%As a result the syndrome measurement reads for $m+1\le i\le n$, 
%\begin{eqnarray}
%R^s_A 
%&=&U^\dagger\left(2^{-(n-m)}\sum_{s'} (-1)^{k\cdot s'}\prod_{i=m+1}^{n} \left(X^{e_{i}}\right)^{s'_{i-m}}\right)U\nonumber\\
%&=&2^{-(n-m)}\sum_{s'} (-1)^{s\cdot s'} U^\dagger X^{(0^{m}\|s')}U\nonumber\\
%&=&
%U^\dagger\sum_{x\in \{0,1\}^m}\ket{\widetilde{x\| s}}\bra{\widetilde{x\| s}}U\nonumber\\
%&=&U^\dagger_A \left(\II_Z\otimes \ket{\widetilde{s}}\bra{\widetilde{s}}_S\right)U_A,
%\end{eqnarray}
%\begin{equation}
%W S^s_A W^\dagger=\II_Z\otimes \ket{\widetilde{s}}\bra{\widetilde{s}}_S.
%\end{equation}
%$C_1^f$ and $C_2^f$ of $PC^g$ become $C'_1=\{0,1\}^n$ and $C'_2=\set{(0^m\|s)|s\in\{0,1\}^{n-m}}$, and accordingly, 
\end{proof}

We transform our virtual PA by applying this $V$ to ${\cal H}_A$ in a straightforward manner:
We transform the initial state as $\Ket{\Psi^{{\rm ini}\prime}}=V_A\Ket{\Psi^{\rm ini}}$, which is also a purification of a cq state with respect to the bit basis. 
%with $V$ of (\ref{eq:U_A_defined}) being a permutation).
%$\rho'_{AE}=V_A\rho_{AE}V_A^\dagger=\sum_{a}\ket{av}\bra{av}_A\otimes \rho_E^{a}=\sum_{a}\ket{a}\bra{a}_A\otimes \rho_E^{av^{-1}}$.
We also replace all the operators $O^i_{AB}\in\left\{(\bar{X}_i)_A, E^{g,s,e}_B, (Z_i)_A, (\bar{Z}_i)_A\right\}$ used for measurement or operation in $\Pi^{{\rm vp},g}$ with $O^{\prime i}_{AB}=V_AO^i_{AB}V_A^\dagger$, and call it the transformed virtual PA algorithm $\Pi^{{\rm vp}\prime,g}$.
We also define the transformed phase error correction $\Pi^{{\rm pec}\prime,g}$ similarly.
We denote the outputs of $\Pi^{{\rm vp}\prime,g}$, $\Pi^{{\rm pec}\prime,g}$ by $\rho^{\text{fin-vp}\prime,g}=\Pi^{{\rm vp}\prime,g}(\ket{\Psi^{{\rm ini}\prime}}\bra{\Psi^{{\rm ini}\prime}})$, $\rho^{{\rm pec}\prime,g}=\Pi^{{\rm pec}\prime,g}(\ket{\Psi'}\bra{\Psi'})$ respectively.

Then it can easily be shown that the statement of lemma \ref{lmm:Pph_upperbound} is invariant under $V_A$:
\begin{proposition}
Quantities on each side of (\ref{eq:lemma1}) are invariant under $V_A$; that is,
$P^{\rm ph}_A\left(\rho^{{\rm pec}\prime,g}\right)=P^{\rm ph}_A\left(\rho^{{\rm pec},g}\right)$ and 
$1-F\left(\rho_{KE}^{{\rm ideal}\prime}, \rho_{KE}^{{\rm fin}\text{-}{\rm vp}\prime, g}\right)^2=1-F\left(\rho_{KE}^{\rm ideal}, \rho_{KE}^{{\rm fin}\text{-}{\rm vp},g}\right)^2$.
\end{proposition}
\begin{proof}
$1-F\left(\rho_{KE}^{\rm ideal}, \rho_{KE}^{{\rm fin}\text{-}{\rm vp}, g}\right)^2$ is invariant because $\rho_{KE}^{{\rm fin}\text{-}{\rm vp}, g}$ is invariant; note here that operation in ${\cal H}_K$ is not changed by the transformation above.
The invariance of $P^{\rm ph}_A\left(\rho^{{\rm pec},g}\right)$ follows by recalling the definition of the failure probability (\ref{eq:failure_prob_PEC_defined}) and by noting $\rho_{KE}^{{\rm pec}\prime,g}=V_A\rho_{KE}^{{\rm pec},g}V_A^\dagger$ and $V_A\ket{\widetilde{0}}_A=\ket{\widetilde{0}}_A$.
\end{proof}
Thus it suffices to prove lemma \ref{lmm:Pph_upperbound} for the transformed case of $\Pi^{{\rm vp}\prime,g}$ and $\Ket{\Psi^{{\rm ini}\prime}}$, and from now on we restrict ourselves to this case.
Note that in this case, measurements in ${\cal H}_A$ are greatly simplified due to (\ref{eq:barZ_transformed_by_V}), (\ref{eq:barX_transformed_by_V}).
Syndrome measurements of step 1(a) are phase basis measurement using the normal Pauli operators $X_i$ of qubits $i=m+1,\dots,n$, and hash value measurements of step 2 are bit basis measurements using $Z_i$ of qubits $i=1,\dots,m$.
Thus it is convenient to rewrite bit string $a\in\{0,1\}^n$ as a concatenation $a=a_K\|a_S$ of $a_K\in\{0,1\}^m$ and $a_S\in\{0,1\}^{n-m}$, and according to that, divide ${\cal H}_A$ as ${\cal H}_A={\cal H}_K\otimes{\cal H}_S$, with
\begin{equation}
\ket{a}_A=\Ket{a_K\|a_S}_A=\ket{a_K}_{K}\otimes\ket{a_S}_{S}.
\label{eq:Space_A_divided}
\end{equation}
Then hash value equals $a_K$, and the syndrome measurement becomes the $X$-basis measurement in ${\cal H}_S$.

\subsubsection{Fixing purification  $\ket{\Psi^{\rm ini}}$}
Next we write down purification $\ket{\Psi^{\rm ini}}$ of $\rho_{AE}$ explicitly.
The form of $\ket{\Psi^{\rm ini}}_{AEB}$ is ambiguous up to unitary transforms in ${\cal H}_B$, but this ambiguity is canceled out when defining ancilla measurement $M_B^{g,s}$, included in the phase error correction $\Pi^{{\rm pec},g}$.
Hence we loose no generality by choosing an arbitrary $\ket{\Psi^{\rm ini}}$.
Here for the ease of notation, we divide ${\cal H}_{A'}$ into ${\cal H}_{A'}={\cal H}_{A_1}\otimes{\cal H}_{A_2}$, and let
\begin{eqnarray}
\ket{\Psi^{\rm ini}}_{AEA'}&=&\sum_{a}\ket{a}_{A}\otimes |\psi^{a}\rangle_{EA_1}\otimes|a\rangle_{A_2}\nonumber\\
&=&\sum_{a_K,a_S}\ket{a_K}_{K}\otimes\ket{a_S}_{S}\otimes \ket{\psi^{a_K\|a_S}}_{EA_1}\nonumber\\
&&\qquad\quad\otimes\ket{a_K}_{K'}\otimes\ket{a_S}_{S'},
\label{eq:purification_Phi_defined}
\end{eqnarray}
where $\ket{\psi^{a}}_{EA_1}$ are arbitrarily chosen purifications of Eve's sub-normalized state $\rho_E^{{\rm ini},a}$ included in the initial state 
(\ref{eq:initial_state_def}), i.e., ${\rm Tr}_{A_1}|\psi^{a}\rangle\langle \psi^{a}|=\rho_E^{a}$, and ${\cal H}_{A_1}$ is an ancilla space introduced for this purpose.
System ${\cal H}_{A_2}$ is a replica of the $n$-qubit system ${\cal H}_A$, which we introduce in order to make $\rho_{AE}={\rm Tr}_{A'}(\ket{\Psi^{\rm ini}}\bra{\Psi^{\rm ini}})$ a cq state. 
In the second line of (\ref{eq:purification_Phi_defined}), we divided ${\cal H}_{A_2}$ further into ${\cal H}_{A_2}={\cal H}_{K'}\otimes{\cal H}_{S'}$, in the same way as in (\ref{eq:Space_A_divided}), as
\begin{equation}
\ket{a}_{A_2}=\Ket{a_K\|a_S}_{A_2}=\ket{a_K}_{K'}\otimes\ket{a_S}_{S'}.
\end{equation}

We note that, in summary, ${\cal H}_{A}$ and ${\cal H}_{A'}$ are divided as
\begin{eqnarray}
{\cal H}_{A}&=&{\cal H}_{K}\otimes{\cal H}_{S},\\
{\cal H}_{A'}&=&{\cal H}_{A_1}\otimes{\cal H}_{A_2}={\cal H}_{A_1}\otimes{\cal H}_{K'}\otimes{\cal H}_{S'}.
\end{eqnarray}

For later convenience, we also define a purification of the ideal state $\rho^{\rm ideal}_{AE}$ to be
\begin{eqnarray}
\ket{\Psi^{\rm ideal}}_{AEA'}&:=&2^{-n/2}
\sum_{a}|a\rangle_{A}\otimes |\varphi\rangle_{EA_1}\otimes|a\rangle_{A_2}\nonumber\\
&=&2^{-n/2}
\sum_{a_K,a_S}\ket{a_K}_{K}\otimes\ket{a_S}_{S}\otimes \ket{\varphi}_{EA_1}\nonumber\\
&&\qquad\qquad\otimes\ket{a_K}_{K'}\otimes\ket{a_S}_{S'},
\end{eqnarray}
where $\ket{\varphi}_{EA_1}$ is a purification of $\rho_E$ defined in (\ref{eq:rho_E_defined}).
Thus it satisfies ${\rm Tr}_{A_1}\ket{\varphi}\bra{\varphi}=\rho_E={\rm Tr}_A\left(\rho_{AE}^{\rm ini}\right)$.

\subsubsection{Unitary operation for phase error correction}
By applying CNOT gate $\Pi^{\rm cnot}_{S'| S}$, defined in (\ref{eq:N_D1_D2_defined}), to these states, we obtain
\begin{align}
&\Pi^{\rm cnot}_{S'| S}\Ket{\Psi^{\rm ini}}_{AEA'}\\
&\ =\sum_{a_K,a_S}|a_K\rangle_{K}\otimes\ket{0}_S\otimes |\psi^{a_K\|a_S}\rangle_{EA_1}\otimes|a_K\rangle_{K'}\otimes|a_S\rangle_{S'},\nonumber\\
&\Pi^{\rm cnot}_{S'| S}\Ket{\Psi^{\rm ideal}}_{AEA'}\\
&\ =2^{-m/2}\sum_{a_K}\ket{a_K}_K\otimes\ket{0}_S\otimes\ket{\varphi}_{EA_1}\otimes\ket{a_K}_{K'}\otimes\ket{\widetilde{0}}_{S'},\nonumber
\end{align}
which are purifications of $\ket{0}\bra{0}_S\otimes\rho_{KE}^{\text{fin-vp},g}$ and $\ket{0}\bra{0}_S\otimes\rho_{KE}^{\rm ideal}$ respectively.
Hence due to Uhlmann's theorem, we can define a unitary operation $T_{A'}$ satisfying
\begin{eqnarray}
\lefteqn{\bra{\Psi^{\rm ideal}} \Pi^{\rm cnot}_{S'| S}T_{A'}\Pi^{\rm cnot}_{S'| S}\Ket{\Psi^{\rm ini}}}\nonumber\\
&=&F\left(\ket{0}\bra{0}_S\otimes\rho^{\rm ideal}_{KE},\ket{0}\bra{0}_S\otimes\rho_{KE}^{\text{fin-vp},g}\right)\nonumber\\
&=&F(\rho^{\rm ideal}_{KE},\rho_{KE}^{\text{fin-vp},g})
\end{eqnarray}
(see, e.g., exercise 2.81 and theorem 9.4 of ref. \cite{Nielsen-Chuang}).
We also note that
\begin{equation}
\ket{\Phi}:=\Pi^{\rm cnot}_{A|A_2}\Ket{\Psi^{\rm ideal}}_{AEA'}=\ket{\widetilde{0}}_A\otimes\ket{\varphi}_{EA_1}\otimes\ket{0}_{A_2}
\end{equation}
is a pure state with zero phase error.
Thus if we define a unitary operation $W_{AA'}$ for phase error correction
\begin{eqnarray}
W_{AA'}&:=&\Pi^{\rm cnot}_{A| A_2}\Pi^{\rm cnot}_{S'| S}T_{A'}\Pi^{\rm cnot}_{S'| S},
\label{eq:define_V_AB}
\end{eqnarray}
its failure probability can be bounded from above as
\begin{eqnarray}
\lefteqn{P_A^{\rm ph}\left(W\ket{\Psi^{\rm ini}}\bra{\Psi^{\rm ini}}W^\dagger\right)}\nonumber\\
&=&1-{\rm Tr}\left\{
\left(\ket{\widetilde{0}}\bra{\widetilde{0}}_A\otimes\II_{EA'}\right)
W\ket{\Psi^{\rm ini}}\bra{\Psi^{\rm ini}}W^{\dagger}
\right\}\nonumber\\
&\le&1-{\rm Tr}\left\{
\left(\Ket{\Phi}\Bra{\Phi}\right)W\ket{\Psi^{\rm ini}}\bra{\Psi^{\rm ini}}W^\dagger
\right\}\nonumber\\
&=&1-
\left|\Bra{\Psi^{\rm ideal}}\Pi^{\rm cnot}_{S'| S}T_{A'}\Pi^{\rm cnot}_{S'| S}\ket{\Psi^{\rm ini}}\right|^2\nonumber\\
%&=&\left|\bra{\Psi^{\rm ideal}}V\ket{\Psi}\right|^2\nonumber\\
&=&1-F(\rho^{\rm ideal}_{KE},\rho_{KE}^{\text{fin-vp},g})^2.
\label{eq:unitary_QEC_BLER}
\end{eqnarray}

\subsubsection{Actual phase error correction}
It remains to show that there exists a phase error correction algorithm of the form of $\Pi^{{\rm pec},g}$ (specified inside $\Pi^{\rm vp}$ of section \ref{sec:obtained_virtual_PA}) that achieves the same failure probability as unitary phase error correction $W_{AA'}$.

To this end, we define an operator set $w$ to be
\begin{equation}
w:=\{\,W^{s,e}_{AA'}\,|\,s\in\{0,1\}^{n-m}, e\in\{0,1\}^n\,\},
%V _{AB}&=&\sum_{s,a}V^{s,a}_{AB},
\end{equation}
where
\begin{equation}
W^{s,e}_{AA'}:=\left(\Ket{\widetilde{s+e_S}}\Bra{\widetilde{s+e_S}}_{S}\otimes\ket{\widetilde{e}}\bra{\widetilde{e}}_{A_2}\right)W_{AA'},
%=u_{A}^{f,s,a}\otimes v_{B}^{f,s,a},\\
%p^a_{B_2}&:=&\ket{\widetilde{a}}\bra{\widetilde{a}}_{B_2},\\
%p^s_S&:=&\ket{\widetilde{s}}\bra{\widetilde{s}}_{S},\\
%u_{A}^{f,s,a}&:=&Z^{a+(0\|s)}_A \left(\II_Z\otimes \ket{\widetilde{s}}\bra{\widetilde{s}}_{S}\right)=Z^{a+(0\|s)}_A q^s_S,\\
%v_{B}^{f,s,a}&:=&\ket{\widetilde{a}}\bra{\widetilde{a}}_{B_2}Z^{s}_{S'}T_BZ^{s}_{S'}=p^a_{B_2}Z^{s}_{S'}T_BZ^{s}_{S'}.
\end{equation}
and $e_S$ denotes the least significant $n-m$ bits of $e$, as in eq. (\ref{eq:Space_A_divided}); $e=e_K\|e_S$.
It is easy to see that $w$ is indeed a quantum operation satisfying
$
\sum_{s,e}W^{s,e\dagger}_{AA'}W^{s,e}_{AA'}=\II_{AA'}
%\label{eq:V_fsa_Kraus_condition}
$.
Further, by using eqs. (\ref{eq:define_V_AB}) and (\ref{eq:N_D1_D2_defined}), and relations $Z^e_A=Z^{e_K}_K\otimes Z^{e_S}_S$ and $
Z^{e_S}_S\ket{\widetilde{e_S+s}}_S=\ket{\widetilde{s}}_S$,  
we can rewrite $W^{s,e}_{AA'}$ as
\begin{eqnarray}
W^{s,e}_{AA'}&=&Z^e_A\ket{\widetilde{s}}\bra{\widetilde{s}}_{S}\otimes E_{A'}^{s,e},\\
E_{A'}^{s,e}&=&\ket{\widetilde{e}}\bra{\widetilde{e}}_{A_2}Z_{S'}^sT_{A'}Z_{S'}^s,
\end{eqnarray}
where $E_{A'}^{s,e}$ satisfy
\begin{equation}
\sum_e E_{A'}^{s,e\dagger}E_{A'}^{s,e}=\II_{A'}.
\end{equation}
Thus the quantum operation $w=\set{W^{s,e}_{AA'}|s,e}$ is indeed of the form of $\Pi^{{\rm pec},g}$.

This operation attains the same failure probability as unitary phase error correction $W$, since
\begin{eqnarray}
\lefteqn{P^{\rm ph}_A\left(\sum_{s,e}W^{s,e}_{AA'}\Ket{\Psi^{\rm ini}}\Bra{\Psi^{\rm ini}}W^{s,e\dagger}_{AA'}\right)}\nonumber\\
&=&
1-\sum_{s,e}
{\rm Tr}
\left\{
\ket{\widetilde{0}}\bra{\widetilde{0}}_A
\left(W^{s,e}_{AA'}\Ket{\Psi^{\rm ini}}\Bra{\Psi^{\rm ini}}W^{s,e\dagger}_{AA'}\right)
\right\}\nonumber\\
&=&
1-\sum_{s,e}
{\rm Tr}
\left\{
\ket{\widetilde{0}}\bra{\widetilde{0}}_A
\left(W^{s+e_S,e}_{AA'}\Ket{\Psi^{\rm ini}}\Bra{\Psi^{\rm ini}}W^{s+e_S,e\dagger}_{AA'}\right)
\right\}\nonumber\\
&=&
1-\sum_{s,e}
{\rm Tr}
\Bigl\{
\left(\left(\ket{\widetilde{s}}\bra{\widetilde{s}}_{S}\ket{\widetilde{0}}\bra{\widetilde{0}}_A\ket{\widetilde{s}}\bra{\widetilde{s}}_{S}\right)\otimes \ket{\widetilde{e}}\bra{\widetilde{e}}_{A_2}\right)\nonumber\\
&&\qquad\qquad\quad\times\left(W_{AB}\Ket{\Psi^{\rm ini}}\Bra{\Psi^{\rm ini}}W^{\dagger}_{AB}\right)
\Bigr\}\nonumber\\
&=&
1-{\rm Tr}
\left\{
\ket{\widetilde{0}}\bra{\widetilde{0}}_A
\left(W_{AA'}\Ket{\Psi^{\rm ini}}\Bra{\Psi^{\rm ini}}W^{\dagger}_{AA'}\right)
\right\}\nonumber\\
&=&P_A^{\rm ph}\left(W_{AA'}\Ket{\Psi^{\rm ini}}\Bra{\Psi^{\rm ini}}W_{AA'}^\dagger\right),
\end{eqnarray}
where we used $\left(\ket{\widetilde{e}}\bra{\widetilde{e}}_{A_2}\right)^2=\ket{\widetilde{e}}\bra{\widetilde{e}}_{A_2}$, $\sum_e\ket{\widetilde{e}}\bra{\widetilde{e}}_{A_2}=\II_{A_2}$.
This completes the proof of lemma \ref{lmm:Pph_upperbound}.

\subsection{Proof of lemma \ref{lmm:universal_hashing}}
\label{sec:proof_purified_dist_upperbound}

In this proof, we abbreviate $\rho^\text{ini-ap}$ and $\rho^{\text{fin-vp},g}$(=$\rho^{\text{fin-ap},g}$), as $\rho^{\rm ini}$ and $\rho^{{\rm fin},g}$, respectively.

According to ref. \cite{TBH14}, there are quantities
\begin{eqnarray}
\widetilde{H}_{1/2}^\downarrow(K|E)_\rho&=&2\log_2 F(\II_K\otimes\rho_E,\rho_{KE}),\\
\widetilde{H}_{2}^\downarrow(K|E)_\rho&=&-\log_2 {\rm Tr}\left\{\left(\rho_{KE}\left(\II_K\otimes \rho_E^{-1/2}\right)\right)^2\right\},\nonumber\\
\end{eqnarray}
which are defined for an arbitrary state $\rho_{KE}$.
These quantities satisfy $\widetilde{H}_{1/2}^\downarrow(K|E)_\rho\ge \widetilde{H}_{2}^\downarrow(K|E)_\rho$, and so the quantity appearing on the right hand side of eq. (\ref{eq:lemma1}) can be bounded as
\begin{eqnarray}
\lefteqn{1-2^{-m}F(\II_K\otimes\rho,\rho_{KE})^2=1-2^{-m+\widetilde{H}_{1/2}^\downarrow(K|E)_\rho}}\nonumber\\
&\le&1-2^{-m+\widetilde{H}_{2}^\downarrow(K|E)_\rho}\le2^{m-\widetilde{H}_{2}^\downarrow(K|E)_\rho}-1\nonumber\\
&=&2^m{\rm Tr}\left\{\left(\rho_{KE}\left(\II_K\otimes \rho_E^{-1/2}\right)\right)^2\right\}-1.
\label{eq:lemma2:1-2mF}
\end{eqnarray}
On the second line of (\ref{eq:lemma2:1-2mF}), we used the fact that $1-1/x\le x-1$ for $ x>0$.

Then by letting $\rho_{KE}=\rho_{KE}^{{\rm fin},g}$ in this inequality and by averaging it with respect to the random choice of universal$_2$ hash function $g$, we obtain
\begin{eqnarray}
\lefteqn{\sum_{g}\Pr(G=g) \left(1-2^{-m}F(\II_K\otimes\rho_E,\rho_{KE}^{{\rm fin},g})^2\right)}\nonumber\\
%&\le&\sum_{g}\Pr(G=g)\,2^{m-\widetilde{H}_{2}^\downarrow(K|E)_{\rho^g}}-1\nonumber\\
&\le&\sum_{g}\Pr(G=g)\,2^m{\rm Tr}\left\{\left(\rho_{KE}^{{\rm fin},g}\left(\II_K\otimes \rho_E^{-1/2}\right)\right)^2\right\}-1\nonumber\\
&=&-1+\sum_{g}\sum_{k}2^m\Pr(G=g)\nonumber\\
&&\times{\rm Tr}\left\{\sum_{\substack{a,a'\\ \in g^{-1}(k)}}(\rho_E^{-1/4}\rho^{{\rm ini},a}_{E}\rho_E^{-1/4})(\rho_E^{-1/4}\rho^{{\rm ini},a'}_{E}\rho_E^{-1/4})\right\}\nonumber\\
&=&-1+\,\sum_{g}\Pr(G=g)2^m\sum_{a,a'}1[g(a)=g(a')]\nonumber\\
&&\qquad\times{\rm Tr}\left\{\left(\rho_E^{-1/4}\rho^{{\rm ini},a}_{E}\rho_E^{-1/4}\right)\left(\rho_E^{-1/4}\rho^{{\rm ini},a'}_{E}\rho_E^{-1/4}\right)\right\}\nonumber\\
&\le&-1+\,2^m\sum_{a,a'}\left(1[a=a']+2^{-m}\right)\nonumber\\
&&\qquad\times
{\rm Tr}\left\{\left(\rho_E^{-1/4}\rho^{{\rm ini},a}_{E}\rho_E^{-1/4}\right)\left(\rho_E^{-1/4}\rho^{{\rm ini},a'}_{E}\rho_E^{-1/4}\right)\right\}\nonumber\\
%&=&
%2^m\sum_x{\rm Tr}\left\{\left(\rho_E^{-1/4}\rho^{x}_{E}\rho_E^{-1/4}\right)^2\right\}
%+{\rm Tr}\left\{\left(\rho_E^{-1/4}\rho_{E}\rho_E^{-1/4}\right)^2\right\}-1\nonumber\\
&=&2^{m-\widetilde{H}_{2}^\downarrow(A|E)_{\rho^{\rm ini}}}.
\label{eq:proof_lemma2_EFTr_rho}
\end{eqnarray}
Here, function $1[\cdots]$ takes value one if the condition inside brackets holds, and zero otherwise.
Due to eq. (\ref{eq:rho_E_defined}), the reduced states for system ${\cal H}_E$ corresponding to $\rho^{{\rm fin},g}$ and to $\rho^{\rm ini}$ are the same, and we denoted it here as $\rho_E$.
The second inequality of (\ref{eq:proof_lemma2_EFTr_rho}) holds because the quantity ${\rm Tr}\left\{\cdots \right\}$ is non-negative, and $g$ is a universal$_2$ hash function, satisfying inequality (\ref{eq:def_universal$_2$_function}).
It should be noted that, in fact, the second and further lines of inequality (\ref{eq:proof_lemma2_EFTr_rho}) are essentially a special case of lemma 5.4.3, ref. \cite{RennerPhD}, with $\sigma_E=\rho_E$.

Next by setting $\alpha=\infty$ in eq. (48), corollary 4, ref. \cite{TBH14} (also see the paragraph of eq. (52) of the same paper), we obtain
\begin{equation}
H_{\min}(\rho_{AE}^{\rm ini}|E)=\widetilde{H}_{\infty}^\uparrow(A|E)_{\rho^{\rm ini}}\le\widetilde{H}_{2}^\downarrow(A|E)_{\rho^{\rm ini}}.
\label{eq:lemma2:H_Alpha_relation}
\end{equation}

Finally by combining inequalities (\ref{eq:proof_lemma2_EFTr_rho}) and (\ref{eq:lemma2:H_Alpha_relation}), we obtain (\ref{eq:lemma2}).

\section{Proof of Inequality (\ref{eq:d_1_bounded_by_P_ph_pre})}
\label{sec:Proof_lemma_phase_error_vs_trace_distance}

Again by using the same argument as in appendix \ref{sec:reduction_to_special_case}, we may assume that system ${\cal H}_A$ is divided as ${\cal H}_A={\cal H}_K\otimes{\cal H}_S$, as in (\ref{eq:Space_A_divided}).
In this case, the step 2.(b) of the virtual PA scheme is the usual $Z$-basis measurement in ${\cal H}_K$.
In addition, it is straightforward to see that phase error probability in ${\cal H}_K$ is bounded by that in ${\cal H}_A$,
\begin{equation}
P^{\rm ph}_{K}\left(\rho^{\text{pec-vp},g}\right)\le P^{\rm ph}_{A}\left(\rho^{\text{pec-vp},g}\right).
\end{equation}
Thus it suffices to prove the following lemma.

\begin{Lmm}
\label{lmm:trace_dist_vs_phase_error}
Let ${\cal H}_K$ be a system of $m$-qubits, and ${\cal H}_E$ be a system of an arbitrary dimension.
Then for an arbitrary sub-normalized state $\rho_{KE}$ in ${\cal H}_{KE}$, we have
\begin{equation}
d_1\left(\Pi^{Z\text{-}{\rm mea}}_K\left(\rho_{KE}\right)\right)\le 2\sqrt2\sqrt{P^{\rm ph}_{K}\left(\rho_{KE}\right)},
\label{eq:d_1_bound_non_smoothed}
\end{equation}
where $\Pi^{Z\text{-}{\rm mea}}_K$ is the $Z$-basis measurement in ${\cal H}_K$; $\Pi^{Z\text{-}{\rm mea}}_K\left(\rho_{KE}\right)=\sum_{z}\ket{z}\bra{z}_K\rho_{KE}\ket{z}\bra{z}_K$ (cf. appendix \ref{sec:notation}).
\end{Lmm}

\begin{proof}
It suffices to bound this lemma for the case where $\rho_{KE}$ is normalized, because the left hand side of (\ref{eq:d_1_bound_non_smoothed}) scales linearly with ${\rm Tr}(\rho_{KE})$, and the right hand side linearly with its square root.

Choose an arbitrary purification $\ket{\Psi'}_{KEB}$ of $\rho_{KE}$, with ${\cal H}_{B}$ being an ancilla space.
Also define another $m$-qubit ancilla space ${\cal H}_{C}$, and let $\ket{\widetilde{0}}_{C}$ be one of the $X$-basis states there.
Then 
\begin{equation}
|\Psi\rangle_{KEBC}=|\Psi'\rangle_{KEB}\otimes\ket{\widetilde{0}}_{C}
\end{equation}
is also a purification of $\rho_{KE}$.
This state can be expanded with respect to the $Z$-basis of ${\cal H}_K$ as
\begin{equation}
|\Psi\rangle_{KEBC}=\sum_{k}|k\rangle_K\otimes|\psi^k\rangle_{EB}\otimes\ket{\widetilde{0}}_{C}.
\end{equation}
With this setting, it is straightforward to verify that $|\psi^k\rangle_{EB}$ and $\Pi^{Z\text{-}{\rm mea}}_K\left(\rho_{KE}\right)$ are related as
\begin{eqnarray}
\Pi^{Z\text{-}{\rm mea}}_K\left(\rho_{KE}\right)&=&\sum\ket{k}\bra{k}_K\otimes \rho^k_E,\\
\rho^k_E&=&{\rm Tr}_{B}(\ket{\psi^k}\bra{\psi^k}_{EB}).
\end{eqnarray}

The ideal state corresponding to $\Pi^\text{$Z$-mea}_K(\rho_{KE})$ takes the form
\begin{eqnarray}
\left(\Pi^\text{$Z$-mea}_K(\rho_{KE})\right)^{\rm ideal}&=&2^{-m}\II_K\otimes\sum_k\rho_E^k,
\end{eqnarray}
and we can choose a semi-purification of $\left(\Pi^\text{$Z$-mea}_K(\rho_{KE})\right)^{\rm ideal}$ to be 
\begin{eqnarray}
%&=&\Pi^\text{$Z$-mea}_K\left({\rm Tr}_{B}\left(|\Psi^{\rm ideal}\rangle\langle\Psi^{\rm ideal}|\right)\right),\\
|\Psi^{\rm ideal}\rangle_{KEBC}&:=&|\tilde{0}\rangle_K\otimes\sum_{k}|\psi^k\rangle_{EB}\otimes|k\rangle_{C};
%\nonumber\\
%&=&|\tilde{0}\rangle_A\otimes\sum_{x}|\widetilde{\phi}^x\rangle_{EB}\otimes|\tilde{x}\rangle_C
\end{eqnarray}
this is equivalent to saying that $\ket{\Psi^{\rm ideal}}$ satisfies  the following relation
\begin{equation}
(\Pi^\text{$Z$-mea}_K(\rho_{KE}))^{\rm ideal}=\Pi^\text{$Z$-mea}_K\left({\rm Tr}_{BC}\ket{\Psi^{\rm ideal}}\bra{\Psi^{\rm ideal}}\right).
\end{equation}

Thus
\begin{eqnarray}
\lefteqn{d_1(\Pi^\text{$Z$-mea}_K(\rho_{KE}))}\nonumber\\
&=&\left\|\Pi^\text{$Z$-mea}_K(\rho_{KE})-(\Pi^\text{$Z$-mea}_K(\rho_{KE}))^{\rm ideal}\right\|_1\nonumber\\
&=&\Bigl\|
\Pi^\text{$Z$-mea}_K\left({\rm Tr}_{BC}\ket{\Psi}\bra{\Psi}\right)\nonumber\\
&&\qquad-\Pi^\text{$Z$-mea}_K\left({\rm Tr}_{BC}\ket{\Psi^{\rm ideal}}\bra{\Psi^{\rm ideal}}\right)\Bigr\|_1\nonumber\\
&\le&\left\|
\ket{\Psi}\bra{\Psi}-\ket{\Psi^{\rm ideal}}\bra{\Psi^{\rm ideal}}
\right\|_1\nonumber\\
%&\le&2\sqrt{1-F(\rho_{AE},\ \rho_{AE}^{\rm ideal})^2}\nonumber\\
&\le&2\sqrt{1-\left|\langle\Psi|\Psi^{\rm ideal}\rangle\right|^2}\nonumber\\
&=&2\sqrt{\left(1-\left|\langle\Psi|\Psi^{\rm ideal}\rangle\right|\right)\left(1+\left|\langle\Psi|\Psi^{\rm ideal}\rangle\right|\right)}\nonumber\\
%&=&2\sqrt{1-(1-P^{\rm ph}(\bar{\rho}_Z))^2}\nonumber\\
&\le&2\sqrt{2}\sqrt{1-\left|\langle\Psi|\Psi^{\rm ideal}\rangle\right|}.
\label{eq:d_1_diag_le_innerprod}
\end{eqnarray}
On the fourth line we used the monotonicity of $L_1$-distance with respect to CPTP maps, $\Pi^\text{$Z$-mea}_K$ and ${\rm Tr}_{BC}$.
The fourth line follows by definitions of the $L_1$ distance and of the fidelity; see, e.g., section 9.2.3, ref. \cite{Nielsen-Chuang}.
The sixth line follows from $\left|\langle\Psi|\Psi^{\rm ideal}\rangle\right|\le 1$.

Next note that $|\Psi\rangle$ and $|\Psi^{\rm ideal}\rangle$ can be rewritten as 
\begin{eqnarray}
|\Psi\rangle
&=&
\sum_{x}|\widetilde{x}\rangle_K\otimes|\widetilde{\psi}^x\rangle_{EB}\otimes\ket{\widetilde{0}}_{C},\\
|\Psi^{\rm ideal}\rangle
&=&
\sum_{x}|\widetilde{0}\rangle_K\otimes|\widetilde{\psi}^x\rangle_{EB}\otimes|\widetilde{x}\rangle_{C},\\
|\widetilde{\psi}^x\rangle_{EB}&:=&2^{-m/2}\sum_{k}(-1)^{x\cdot k}|\psi^k\rangle_{EB},
\end{eqnarray}
and thus $P^{\rm ph}_K\left(\rho_{KE}\right)$ takes the form
\begin{equation}
P^{\rm ph}_{K}\left(\rho_{KE}\right)=1-\langle \widetilde{\psi}^0 |\widetilde{\psi}^0\rangle_{EB}=1-\left|\langle\Psi|\Psi^{\rm ideal}\rangle\right|.
\label{eq:phase_error_innerprod}
\end{equation}
By combining (\ref{eq:d_1_diag_le_innerprod}) and (\ref{eq:phase_error_innerprod}), we obtain (\ref{eq:d_1_bound_non_smoothed}).

\end{proof}

\section{Proof of lemma \ref{lmm:security_reduction_to_alice}}
\label{sec:proof_separation_of_trace_distance_and_BEC}

In this proof, we often abbreviate $\rho_{KK'EG}^\text{fin-aq}$ as $\rho_{KK'EG}^{\rm fin}$.
By using the triangle inequality of the $L_1$-distance, $D_1(\rho_{KK'EG}^{\rm fin})$ can be bounded as
\begin{eqnarray}
\lefteqn{D_1(\rho_{KK'EG}^{\rm fin})}\nonumber\\
&=&\left\|\rho_{KK'EG}^{\rm fin}-\rho_{KK'E}^{\rm ideal}\right\|_1
\label{eq:triangle_inequality_applied_to_D_1}\\
&\le&\left\|\rho_{KK'EG}^{\rm fin}-\rho^{\rm int}_{KK'EG}\right\|_1+\left\|\rho^{\rm int}_{KK'EG}-\rho_{KK'EG}^{\rm ideal}\right\|_1.
\nonumber
\end{eqnarray}
where $\rho^{\rm int}_{KK'EG}$ is defined as
\begin{eqnarray}
\rho^{\rm int}_{KK'EG}&:=&\sum_{k}\ket{k}\bra{k}_K\otimes\ket{k}\bra{k}_{K'}\otimes\left(\sum_{k'}\rho^{{\rm fin},k,k'}_{EG}\right),\nonumber\\
\\
\rho^{k,k'}_{EG}&=&\sum_{g}\Pr(G=g)\rho^{{\rm fin},g,k,k'}_{E}\otimes\ket{g}\bra{g}_G
\end{eqnarray}
with $\rho^{{\rm fin},g,k,k'}_{E}$ being $\rho^{\text{fin-aq},g,k,k'}_{E}$ defined in (\ref{eq:rho_fin_gkk_defined}).
Superscript `int' of $\rho^{\rm int}$ means that it is an `intermediate' state between $\rho_{KK'EG}^{\rm fin}$ and $\rho_{KK'EG}^{\rm ideal}$.

The two terms on the right hand side of (\ref{eq:triangle_inequality_applied_to_D_1}) can each be bounded as
\begin{eqnarray}
\lefteqn{\left\|\rho_{KK'EG}^{\rm fin}-\rho^{\rm int}_{KK'EG}\right\|_1}\nonumber\\
&=&\biggl\|\sum_{k, k'}\ket{k}\bra{k}_K\otimes\ket{k'}\bra{k'}_{K'}\nonumber\\
&&\qquad\otimes\left(\rho^{{\rm fin},k,k'}_{EG}-\delta_{k,k'}\left(\sum_{k'}\rho^{{\rm fin},k,k'}_{EG}\right)\right)\biggr\|\nonumber\\
&=&\sum_{k\ne k'}\left\|\rho^{{\rm fin},k,k'}_{EG}\right\|
+\sum_{k}\left\|\rho^{{\rm fin},k,k}_{EG}-\sum_{k'}\rho^{{\rm fin},k,k'}_{EG}\right\|
\nonumber\\
&=&\sum_{k\ne k'}\left\|\rho^{{\rm fin},k,k'}_{EG}\right\|
+\sum_{k}\left\|\sum_{k'(\ne k)}\rho^{{\rm fin},k,k'}_E\right\|
\nonumber\\
&\le&2\sum_{k\ne k'}\left\|\rho^{{\rm fin},k,k'}_{EG}\right\|=2\sum_{k\ne k'}{\rm Tr}\left(\rho^{{\rm fin},k,k'}_{EG}\right)\nonumber\\
&=&2\Pr(K\ne K'\,\land\,\rho_{KK'}^{\rm fin}),
\label{eq:appD:term1}
\end{eqnarray}
and
\begin{eqnarray}
\lefteqn{\left\|\rho^{\rm int}_{KK'EG}-\rho_{KK'EG}^{\rm ideal}\right\|_1}\nonumber\\
&=&
\left\|
\sum_{k}\ket{k}\bra{k}_K\otimes\ket{k}\bra{k}_{K'}\otimes
\left(\sum_{k'}\rho^{{\rm fin},k,k'}_{EG}-2^{-m}\rho^{\rm fin}_{EG}\right)\right\|\nonumber\\
&=&
\left\|
\sum_{k}\ket{k}\bra{k}_K\otimes\left(\sum_{k'}\rho^{{\rm fin},k,k'}_{EG}-2^{-m}\rho_{EG}^{\rm fin}\right)\right\|
\nonumber\\
&=&\left\|{\rm Tr}_{K'}(\rho_{KK'EG}^{\rm fin})-2^{-m}\II_K\otimes\rho_{EG}^{\rm fin}\right\|_1\nonumber\\
&=&d_1(\rho_{KEG}^{\rm fin}).
\label{eq:appD:term2}
\end{eqnarray}
Combining Inequalities (\ref{eq:triangle_inequality_applied_to_D_1}), (\ref{eq:appD:term1}), and (\ref{eq:appD:term2}), we obtain inequality (\ref{eq:separation_secrecy_correctness}).

\section{Proof of lemma \ref{lmm:duality}}
\label{sec:proof_duality}

Let $\rho_{A\bar{A}E}^\text{pre-vq3,$Z$-mea}$ be the result of applying the $Z$-basis measurement in system ${\cal H}_A$ on $\ket{\Psi^\text{pre-vq3}}_{A\bar{A}E}$,
\begin{eqnarray}
\lefteqn{\rho_{A\bar{A}E}^\text{pre-vq3,$Z$-mea}}\nonumber\\
&=&\Pi^\text{$Z$-mea}_A(\ket{\Psi^\text{pre-vq3}}\bra{\Psi^\text{pre-vq3}}_{A\bar{A}E})
\label{eq:rho_sif_zmea}\\
&=&\sum_{z}\ket{z}\bra{z}_A\left(\ket{\Psi^\text{pre-vq3}}\bra{\Psi^\text{pre-vq3}}_{A\bar{A}E}\right)\ket{z}\bra{z}_A.\nonumber
\end{eqnarray}
Simlarly,  let $\rho_{A\bar{A}E}^\text{pre-vq3,$X$-mea}$ be the result of applying the $X$-basis measurement in system ${\cal H}_A$ on $\ket{\Psi^\text{pre-vq3}}_{A\bar{A}E}$,
\begin{eqnarray}
\lefteqn{\rho_{A\bar{A}E}^\text{pre-vq3,$X$-mea}}\nonumber\\
&=&\Pi^\text{$X$-mea}_A(\ket{\Psi^\text{pre-vq3}}\bra{\Psi^\text{pre-vq3}}_{A\bar{A}E})
\label{eq:rho_sif_xmea}\\
&=&\sum_{x}\ket{\widetilde{x}}\bra{\widetilde{x}}_A\left(\ket{\Psi^\text{pre-vq3}}\bra{\Psi^\text{pre-vq3}}_{A\bar{A}E}\right)\ket{\widetilde{x}}\bra{\widetilde{x}}_A
\nonumber
\end{eqnarray}
Then by using the uncertainty relation between the conditional min- and max- entropies (theorem 7.1, ref. \cite{TomamichelPhD}), we obtain
\begin{equation}
H_{\rm min}(\rho_{AE}^\text{pre-vq3,$Z$-mea}|E)+H_{\rm max}(\rho_{A\bar{A}}^\text{pre-vq3,$X$-mea}|\bar{A})\ge n,
\label{eq:duality_inequality}
\end{equation}
where $\rho_{AE}^\text{pre-vq3,$Z$-mea}={\rm Tr}_{\bar{A}}(\rho_{A\bar{A}E}^\text{pre-vq3,$Z$-mea})$, and 
$\rho_{A\bar{A}}^\text{pre-vq3,$X$-mea}={\rm Tr}_{E}(\rho_{A\bar{A}E}^\text{pre-vq3,$X$-mea})$.

Next note that $\ket{\Psi^\text{pre-vq3}}_{A\bar{A}E}$ is a semi-purification of $\rho^\text{sif-aq}_{AE}$ due to lemma \ref{lmm:virtualQKD3}.
Then it follows that
\begin{equation}
\rho_{AE}^\text{pre-vq3,$Z$-mea}=\rho^\text{sif-aq}_{AE}
\end{equation}
and thus
\begin{equation}
H_{\rm min}(\rho_{AE}^\text{pre-vq3,$Z$-mea}|E)=H_{\rm min}(\rho^\text{sif-aq}_{AE}|E)
\label{eq:Hminrhoprevq3equalHminrhosif}
\end{equation}

Also by using the data processing inequality of the maximum entropy \cite{TomamichelPhD},
$H_{\rm max}(\rho_{A\bar{A}}^\text{pre-vq3,$X$-mea}|\bar{A})$ can be bounded as
\begin{equation}
H_{\rm max}(\rho_{A}^\text{pre-vq3,$X$-mea})\ge H_{\rm max}(\rho_{A\bar{A}}^\text{pre-vq3,$X$-mea}|\bar{A}),
\label{eq:H_max_data_processing}
\end{equation}
where $\rho_{A}^\text{pre-vq3,$X$-mea}={\rm Tr}_{\bar{A}}(\rho_{A\bar{A}}^\text{pre-vq3,$X$-mea})$.
More precisely, inequality (\ref{eq:H_max_data_processing}) follows from Result 4, section 5.1.1, ref. \cite{TomamichelPhD}, which is a special case of theorem 5.7 of the same literature, with ${\cal H}_{B'}$, ${\cal H}_C$ and ${\cal H}_{C'}$ being  one-dimensional spaces, and maps ${\cal E}$ and ${\cal F}$ being the identity, and  ${\cal G}$ being tracing out of ${\cal H}_{B}$.

Further, due to (\ref{eq:Pr_X_prevq3_defined}), $\rho_{A}^\text{pre-vq3,$X$-mea}$ takes the form
\begin{equation}
\rho_{A}^\text{pre-vq3,$X$-mea}=\sum_{x}\Pr(X^\text{pre-vq3}=x)\ket{\widetilde{x}}\bra{\widetilde{x}}_A.
\end{equation}
Thus by using the results of the section 4.3.2, ref. \cite{TomamichelPhD}, it follows that
\begin{equation}
H_{1/2}(X^\text{pre-vq3})=H_{\rm max}(\rho_{A}^\text{pre-vq3,$X$-mea}).
\label{eq:renyi_equal_max_ent}
\end{equation}

Finally combining (\ref{eq:bound_H_12X}), (\ref{eq:duality_inequality}), (\ref{eq:Hminrhoprevq3equalHminrhosif}), (\ref{eq:H_max_data_processing}), and (\ref{eq:renyi_equal_max_ent}), we obtain the lemma.

\bibliographystyle{IEEEtran}
\bibliography{minentropy}

\begin{IEEEbiographynophoto}{Toyohiro Tsurumaru} was born in Japan in 1973.
He received the B.S. degree from the Faculty of Science, University of Tokyo, Japan in 1996,
and the M.S. and Ph.D. degrees in physics from the Graduate School of Science, University of Tokyo, Japan in 1998 and 2001, respectively.
Then he joined Mitsubishi Electric Corporation in 2001.
His research interests include theoretical aspects of quantum cryptography, as well as modern cryptography.
\end{IEEEbiographynophoto}

\end{document}